\documentclass[final,3p,times]{elsarticle}

\usepackage{amssymb}   
\usepackage{caption}   
\usepackage{subfigure} 
\usepackage{amsmath}   
\usepackage{setspace}  
\usepackage{algorithm} 
\usepackage{algpseudocode} 
\usepackage{makecell}  
\usepackage{multirow}  
\usepackage[justification=centering]{caption} 
\usepackage{multicol}
\usepackage{stfloats}  
\usepackage{color}     
\usepackage{amsthm}    
\usepackage{threeparttable}

\usepackage{dcolumn}

\usepackage{algorithm}
\usepackage{algpseudocode}

\newtheorem{definition}{Definition}
\newtheorem{example}{Example}
\newtheorem{theorem}{Theorem}

\begin{document}
\begin{frontmatter} 
%
\title{OPR-Miner: Order-preserving rule mining for time series}
%
%
%
	\author[1,2]{Youxi Wu}

	\author[1]{Xiaoqian Zhao}
	
	\author[3]{{Yan Li}\corref{mycorrespondingauthor}}
	\ead{lywuc@163.com}
	
	\author[4]{Lei Guo}
	
	\author[5]{Xingquan Zhu}
	
	\author[6]{Philippe Fournier-Viger}
	
	\author[7]{Xindong Wu}
	\cortext[mycorrespondingauthor]{Corresponding author}
	
	\address[1]{School of Artificial Intelligence, Hebei University of Technology, Tianjin 300401, China}
	
	\address[2]{Hebei Key Laboratory of Big Data Computing, Tianjin {\rm 300401}, China}
	
	\address[3]{School of Economics and Management, Hebei University of Technology, Tianjin 300401, China}
	
	\address[4]{State Key Laboratory of Reliability and Intelligence of Electrical Equipment, Hebei University of Technology, Tianjin 300401, China}
	
	\address[5]{The Department of Computer \& Electrical Engineering and
		Computer Science, Florida Atlantic University, FL 33431, USA} 
	
	\address[6]{Shenzhen University, Shenzhen, China} 
	
	\address[7]{Key Laboratory of Knowledge Engineering with Big Data (the Ministry of Education), Hefei University of Technology, Hefei 230009, China} 

\begin{abstract}
Discovering frequent trends in time series is a critical task in data mining. Recently, order-preserving matching was proposed to find all occurrences of a pattern in a time series, where the pattern is a relative order (regarded as a trend) and an occurrence is a sub-time series whose relative order coincides with the pattern. Inspired by the order-preserving matching, the existing order-preserving pattern (OPP) mining algorithm employs order-preserving matching to calculate the support, which leads to low efficiency. To address this deficiency, this paper proposes an algorithm called efficient frequent OPP miner (EFO-Miner) to find all frequent OPPs. EFO-Miner is composed of four parts: a pattern fusion strategy to generate candidate patterns, a matching process for the results of sub-patterns to calculate the support of super-patterns, a screening strategy to dynamically reduce the size of prefix and suffix arrays, and a pruning strategy to further dynamically prune candidate patterns. Moreover, this paper explores the order-preserving rule (OPR) mining and proposes an algorithm called OPR-Miner to discover strong rules from all frequent OPPs using EFO-Miner. Experimental results verify that OPR-Miner gives better performance than other competitive algorithms. More importantly, clustering and classification experiments further validate that OPR-Miner achieves good performance.

\end{abstract}


\end{frontmatter}

%

\section{Introduction}
%
%
%
%

A time series is a continuous numerical series  of data or a group of real values that is commonly used in many fields, such as brain EEG clustering \cite{Dai2022tcyb}, stock prediction \cite{li2021tkde}, and weather forecasting \cite{karevan2020}.  Many studies have been investigated. For example, Wu and Keogh \cite{wukeogh} focused on time series anomaly detection. Rezvani et al. \cite{2021tkde} studied a new pattern representation method for time series data to effectively detect the change point. Sequential pattern mining method, as a commonly used method, can also be used to discover patterns of interest to users in time series \cite{wu2021tmis} after  discretizing the time series into symbols.  {Note that although in episode mining, an event sequence has a set of consecutive time stamps \cite{kdd1996, philippe2019}, it is far different from time series, since an event sequence is a group of discrete events, while time series is a group of continuous numerical values. Therefore, users can directly apply the episode mining methods on event sequences \cite{aotkde2017,icdm2021}, while users have to adopt some discretization methods at first, and then apply some sequential pattern mining methods on time sequence.}

However, the existing discretizing methods pay too much attention to the values, such as piecewise linear approximation (PAA) {\cite{PAA}} and symbolic aggregate approximation (SAX) {\cite{SAX}}. Therefore, it is difficult to discover the frequent trends using sequential pattern mining methods. To address this deficiency, several methods have been investigated to find subsequences with the same trend, such as (delta, gamma) approximate matching {\cite{li2021apind, wu2020netd}}, weak gap strong pattern mining {\cite{wu2021insweak}}, and tri-way pattern mining \cite{Min2020ins, wu2021ntpm}. These methods need to set the parameters manually, which may cause the loss of important information in the process and destroy the continuity of the time series.

Recently, order-preserving matching \cite{kim2014orde,cho2013fast} (or called consecutive permutation pattern matching \cite{ipl2013}) has been proposed, which does not need to discretize real numbers into symbols. Order preserving matching can find all occurrences of a pattern in a time series, where the pattern is a relative order (regarded as a trend) and an occurrence is a sub-time series whose relative order coincides with the pattern. Inspired by order-preserving  matching, our previous work proposed the order-preserving pattern mining (OPP-Miner) algorithm {\cite{wu2021orde}}, which used the relative order of real values to express a pattern called an order-preserving pattern (OPP). By mining OPPs, we can find frequent trends in a time series.  An illustrative example is shown in Fig. \ref{fig:aa}.  In the figure, regions A, B, and C have different means and variances, and the means of A, B, and C are 30.25, 23.25, and 24.75, respectively. The variances of A, B, and C are 4.69, 2.19, and 3.69, respectively. Finding patterns from such non-stationary data is challenging, because of the changing mean and variance. On the other hand, patterns may continually repeat themselves but with different mean values.  For example, over years, the stock index has increased many times (showing an increased mean value), whereas the market patterns are rather similar. By observing order of patterns within a local region, OPP mining can find repetitive patterns with different mean values.

\begin{figure}
	\centering
	\includegraphics[width=0.9\linewidth]{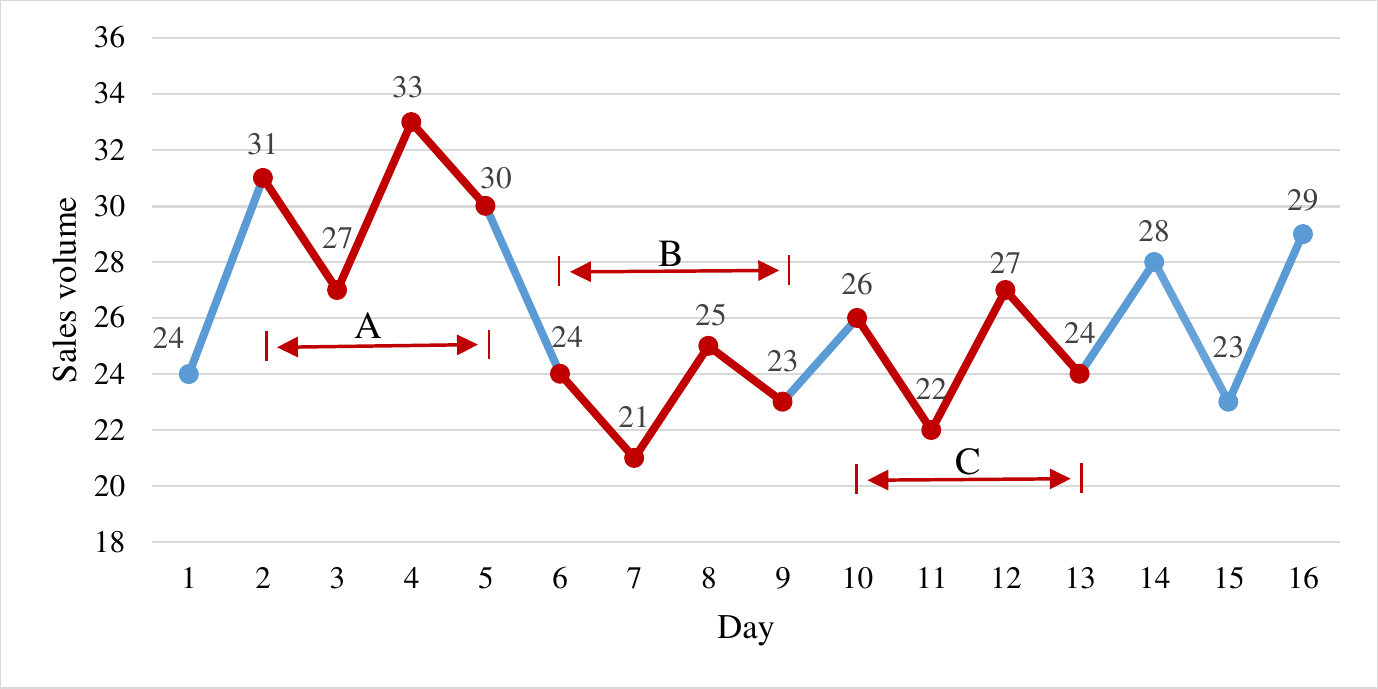}
	\caption{Sales volume of goods over 16 days. The relative order of sub-time series $(t_2,t_3,t_4,t_5)$=(31,27,33,30) is (3,1,4,2), since 31 is the third smallest, 27 is the smallest,  and so on. It can be seen that the trends in the sub-time series marked in red are exactly the same. If $minsup$ = 3, then (3,1,4,2) is a frequent OPP, and OPP mining can discover similar frequently occurring trends.  }
	\label{fig:aa}
\end{figure}

However, there are two problems with OPP mining: (i) OPP-Miner {\cite{wu2021orde}} adopts a pattern matching method to calculate pattern support. Although the space complexity of OPP-Miner is low, its efficiency is also low, since it does not use the calculation results of the sub-patterns. Hence, the efficiency of OPP-Miner needs to be improved. (ii) More importantly, although all OPPs can be discovered, how to further apply these mining patterns has not been deeply explored.

To improve the performance of OPP-Miner, we propose an algorithm called efficient frequent order-preserving pattern miner (EFO-Miner). Moreover, to  utilize these OPPs effectively, we develop order-preserving rule (OPR) mining and propose the OPR-Miner algorithm, which can mine the implicit relationships between OPPs. The main contributions of the paper are as follows.

1) To efficiently mine frequent OPPs, we propose an EFO-Miner algorithm, which employs four strategies: pattern fusion, support-based pattern fusion, screening, and pruning.

2) To mine the implicit relationships between OPPs, we  propose the OPR-Miner algorithm based on EFO-Miner to discover strong rules.

3) Experimental results verify that OPR-Miner yields better performance than other competitive algorithms. Moreover, clustering and classification experiments validate that OPR-Miner can be used to realize feature extraction and achieve good performance.

The rest of this paper is organized as follows. Section \ref {section2} introduces related work. Section  \ref {section3} provides a definition of the problem. Section  \ref {section4} proposes the OPR-Miner algorithm and presents an analysis of its time and space complexities. In Section  \ref {section5}, we validate the performance of OPR-Miner. Section  \ref {section6} concludes this paper.

\section{RELATED WORK} \label {section2}

Sequential pattern mining  {\cite{fournierviger2014spmf}} is an important topic in the field of data mining, whose aim is to mine the subsequences  from a sequential dataset that users are interested in and to help people understand the data and make decisions by analyzing the potential patterns \cite {Okolica2020tkde}. {To solve different types of problems, sequential pattern mining has been extended to include a variety of mining methods, such as sequential pattern mining with gap constraints (or repetitive sequential pattern mining) \cite {wuapin2014},  negative sequential pattern mining \cite{dong2019mini, wunegativetkdd}, high utility pattern mining \cite{gan2021tkde}, high average-utility pattern mining \cite {weisong2021,Truong2019},  episode mining \cite{ Mannila1997, aoicde2015}, and OPP mining for time series \cite {wu2021orde}. }

Variours sequential pattern mining methods have been applied in many fields, such as  disease prediction {\cite{ghosh2017sept}}, virus sequence analysis {\cite{li2021apinm}}, and network clickstream analysis {\cite{nish2018alat}}. For example,  Duan et al. {\cite{duan2020tkdd}} used outlying sequence pattern mining to analyze the outliers in sequence data. {Wu et al. {\cite{wu2019mini}} developed top-$k$ contrast pattern mining to realize the feature extraction of sequence classification.} Smedt et al. \cite{smedt2020tkde} discovered patterns for sequence classification using behavioral constraint templates.   Wu et al. {\cite{wu2018nose}} used a Nettree to calculate the support of a pattern under nonoverlapping conditions. Zhang et al. {\cite{zhang2017onap}} proposed a sequential pattern mining method based on periodic gap constraints. 

However, frequent pattern mining may ignore the implicit relationships within the transaction, and sequential rule mining {\cite{fournierviger2012cmru}} was proposed to address this problem. For example, Pham et al. {\cite{pham2014anef}} proposed an efficient method of mining sequential rules by constructing a prefix tree structure, which generated a large number of redundant rules in the process.   Moreover, Fournier-Viger et al. {\cite{fournierviger2015tkde}} proposed a partially-ordered sequential rule mining to improve prediction accuracy.

Although the works described above have achieved good mining results, these studies mainly focused on the mining of discrete sequences, such as DNA or protein sequences. Due to the high continuity of time series, it is difficult to apply this approach to time series composed of ordered and continuous values. A classical way is that users employ the symbolization methods to discretize the original real values into symbols, and then apply the sequential pattern mining methods to find the interesting patterns. Typical symbolization algorithms include segmentation notation, represented by PAA {\cite{PAA}}, and symbolic representation, represented by SAX {\cite{SAX}}. The main advantage of the time series symbolization method is that the time series is converted into a sequence of symbols through certain transformation rules, thus allowing traditional symbol sequence mining methods to be applied. However, various kinds of noise are inevitably introduced, due to the setting of various hard intervals in the process of converting time series into symbol series. In addition, these methods also ignore the original characteristics of the sequence, making it difficult to find the trends in the data.

To overcome the drawbacks of the symbolization methods, our previous work proposed the OPP mining method which does not need to symbolize the time series \cite{wu2021orde}. To effectively discover the frequent OPPs, OPP-Miner was proposed and employed an OPP matching method to calculate the supports.  In terms of OPP matching, Kim et al. {\cite{kim2014orde}} employed the KMP algorithm to find subsequences with the same trend in a sequence. However, their approach did not consider the case of equal values, and Cho et al. {\cite{cho2013fast}} therefore designed a new algorithm to determine whether two time series were in the same order, even if some elements were equal. To further improve the matching efficiency, Chhabra and Tarhio \cite{Chhabra2016ipl} proposed a filtration method to find all  occurrences.

{However, OPP-Miner \cite{wu2021orde} has two drawbacks. Firstly, the efficiency of OPP-Miner can be further improved, since OPP-Miner adopts a pattern matching method to calculate pattern support, which does not use the calculation results of the sub-patterns. Secondly, OPP-Miner discovers all OPPs. Nevertheless, the implicit relationships between OPPs are not discovered. To overcome the drawbacks of OPP-Miner, this paper proposes the EFO-Miner algorithm, which utilizes the results from sub-patterns to calculate the support of super-patterns, in order to effectively avoid redundant calculations and improve the mining efficiency. More importantly, this paper further proposes the OPR-Miner algorithm based on the EFO-Miner algorithm to find strong rules which can discover the implicit relationships between OPPs, and can be used to extract time series features for clustering and classification.}

\section{Problem Definition} \label {section3}

\begin{definition}\label{definition1}
	A time series is a numerical series of the same statistical indicator that is arranged in the order of its occurrence time, and is denoted as \textbf{t} = ($t_1,\ldots,t_i,\ldots,t_n$), where 1$\leqslant i\leqslant n$.
\end{definition}

\begin{definition}\label{definition2}
	The rank of an element $p_{i}$ in pattern \textbf{p} = $(p_{1},\ldots,p_{i},\ldots,p_{m})$ (1$\leqslant i\leqslant m$) is denoted as $rank_{\textbf{p}}(p_{i}).$ A pattern represented by the relative order of the elements is called an OPP, and can be expressed as $R(\textbf{p}) = (rank_{\textbf{p}}(p_{1}),rank_{\textbf{p}}(p_{2}),\ldots,rank_{\textbf{p}}(p_{m}))$.
\end{definition}

\begin{example}
	Suppose we have a pattern \textbf{p} = (31,27,33,30). We know that 31 is the third smallest value in \textbf{p}, i.e., \textit{rank}(31) = 3. Similarly, \textit{rank}(27) = 1. Thus, the OPP of \textbf{p} is  \textit{R}(\textbf{p}) = (3,1,4,2).
\end{example}

\begin{definition}\label{definition3}
	{Suppose we have a pattern $\textbf{p} = (p_{1},p_{2},\ldots,p_{m})$ and a time series \textbf{t} = ($t_1,\ldots,t_i,\ldots,t_n$). If there exists a sub-time series  \textbf{t$'$ }= $(t_{i},t_{i+1},\ldots,t_{i+m-1})$ (1$\leqslant i $ and $i+m$-$1\leqslant n)$ which satisfies \textit{R}(\textbf{t$'$}) = \textit{R}(\textbf{p}), then \textbf{t$'$} is an occurrence of pattern \textbf{p} in time series \textbf{t}, and we use $<$\textit{i}+\textit{m}$-$1$>$   to represent the occurrence. The support of \textbf{p} in \textbf{t} is the number of occurrences, denoted by \textit{sup}(\textbf{p}, \textbf{t}).}
\end{definition}

\begin{definition}\label{definition4}
	Given a minimum support threshold \textit{minsup}, if the support of \textbf{ p} in \textbf{t} is no less than \textit{minsup}, i.e., \textit{sup}(\textbf{p}, \textbf{t}) $\geqslant$ \textit{minsup}, then pattern \textbf{p} is called a frequent OPP.
\end{definition}

\begin{example}\label{example2}
	Suppose we have a sequence \textbf{t} = (24,31,27,33,30,24,21,25,23,26,22,27,24,28,23,29), as shown in Fig. \ref{fig:aa}, and a sub-time series $(t_{2},t_{3},t_{4},t_{5})$ = (31,27,33,30). We know that \textbf{p} = $R(t_{2},t_{3},t_{4},t_{5})$ = (3,1,4,2). Similarly, $R(t_{6},t_{7},t_{8},t_{9})$ = $R(t_{10},t_{11},t_{12},t_{13}) $= (3,1,4,2). There are therefore three occurrences of pattern (3,1,4,2) in t, i.e., \textit{sup}(\textbf{p}, \textbf{t}) = 3. If\textit{ minsup} = 3, then pattern \textbf{p} is a frequent OPP. In this way, we can get all frequent OPPs \textit{F} = \{(1,2), (2,1), (1,3,2), (2,1,3), (1,3,2,4), (3,1,4,2)\}.
\end{example}

\begin{definition}\label{definition5}
	Given a pattern \textbf{p} = $(p_{1},p_{2},\ldots,p_{m})$, the sub-time series \textbf{e} =$ R(p_{1},p_{2},\ldots,p_{m-1})$ is called the prefix OPP of \textbf{p}, and is denoted as \textbf{e} = \textit{prefix}(\textbf{p}). Sub-time series \textbf{k} = $\textit{R}(p_{2},p_{3},\ldots,p_{m})$ is called the suffix OPP of \textbf{p}, and is denoted as \textbf{k} = \textit{suffix}(\textbf{p}), where \textbf{e} and \textbf{k} are the order-preserving sub-patterns of \textbf{p}, and \textbf{p} is the order-preserving super-pattern of \textbf{e} and \textbf{k}.
\end{definition}

\begin{definition}\label{definition6}
	Suppose \textbf{x} and \textbf{y} are frequent OPPs. If\textbf{ x} is the prefix OPP of \textbf{y}, then the implication \textbf{x}$\to$\textbf{y} is called an order-preserving rule, where \textbf{x} is the antecedent of the rule, and \textbf{y} is the consequent of the rule.
\end{definition}

\begin{definition}\label{definition7}
	The confidence rate of \textbf{x}$\to$\textbf{y}, denoted as \textit{conf}(\textbf{x}$\to$\textbf{y}), is the ratio of the support of\textbf{ y} to that of \textbf{x}, i.e., \textit{conf}(\textbf{x}$\to$\textbf{y}) = \textit{sup}(\textbf{y}, \textbf{t})/\textit{sup}(\textbf{x}, \textbf{t}).
\end{definition}

\begin{definition}	\label{definition8}
	Given a minimum confidence rate threshold \textit{minconf}, if \textit{conf}(\textbf{x}$\to$\textbf{y}) $\geqslant$\textit{ minconf}, then  \textbf{x}$\to$\textbf{y} is called a strong OPR.
\end{definition}

\begin{definition}\label{definition9}	
	Our aim is to discover all strong OPRs in frequent OPPs according to \textit{minconf}.
\end{definition}

\begin{example}\label{example3}
	In Example 2, (2,1,3) is the prefix OPP of (3,1,4,2). Both (2,1,3) and (3,1,4,2) are frequent OPPs, and their supports are 4 and 3, respectively. Hence, \textit{conf}((2,1,3)$\to$(3,1,4,2)) = 3/4 = 0.75. If \textit{minconf} = 0.7, then (2,1,3)$\to$(3,1,4,2) is a strong OPR. Since \textit{conf}((2,1)$\to$(2,1,3)) = 4/8 = 0.5, which is less than \textit{minconf}, it is not a strong OPR. The strong OPRs in Example 2 are \textit{R} = \{(1,2)$\to$(1,3,2), (2,1,3)$\to$(3,1,4,2)\}.
\end{example}

\section{Proposed algorithms} \label {section4}

In OPR mining, the key issue is to discover frequent OPPs. In Section \ref{sub4.1}, we introduce the principle of pattern fusion to generate candidate patterns. We propose the methods of support calculation based on pattern fusion (SPF) and SPF-Pro in  Sections \ref{sub4.2} and \ref{sub4.22}, respectively. Section \ref{sub4.3} illustrates the pruning strategy that is applied to further prune candidate patterns based on SPF-Pro. Section \ref{sub4.4} presents EFO-Miner, which is used to mine frequent OPPs. Finally, Section \ref{sub4.5} proposes OPR-Miner to discover strong rules.

\subsection{Generating candidate patterns }\label {sub4.1}
To reduce the number of candidate patterns, we adopt a pattern fusion method proposed in \cite{wu2021orde} to generate candidate patterns.

For \textbf{p} = $(p_{1},p_{2},\ldots,p_{m})$ and \textbf{q} = $(q_{1},q_{2},\ldots,q_{m})$, where \textit{m} is the length of the patterns, if \textit{R}(\textit{suffix}(\textbf{p})) = \textit{R}(\textit{prefix}(\textbf{q})), then \textbf{p} and \textbf{q} can generate a super-pattern with length \textit{m}+1.  Two cases are given below:

\textbf{Case 1}: If $p_{1} \neq q_{m}$, then \textbf{p} and \textbf{q} can generate one pattern \textbf{r} = $(r_{1},r_{2},\ldots,r_{m+1})$, denoted as \textbf{r} = \textbf{p} $\bigoplus$ \textbf{q}.

1. If $p_{1} < q_{m}$, then $r_{1} = p_{1}$. Moreover, if $q_{i} < p_{1}$, then $r_{i+1} = q_{i}$. Otherwise, $r_{i+1} = q_{i}+1 (1< i \leqslant m)$.

2. If $p_{1} > q_{m}$, then $r_{1} = p_{1}+1$. Moreover, if $q_{i} < p_{1}$, then $r_{i+1} = q_{i}$. Otherwise, $r_{i+1} = q_{i}+1 (1< i \leqslant m)$.

\textbf{Case 2}: If $p_{1} = q_{m}$, then \textbf{p} and \textbf{q} can generate two patterns \textbf{r} = $(r_{1},r_{2},\ldots,r_{m+1})$ and \textbf{h} = $(h_{1},h_{2},\ldots,h_{m+1})$ , denoted as \textbf{r}, \textbf{h} = \textbf{p} $\bigoplus$ \textbf{q}. For pattern \textbf{r}, $r_{1} = p_{1}$ and $r_{m+1} = p_{1}+1$. Moreover, if $q_{i} < p_{1}$, then $r_{i+1} = q_{i}$. If $q_{i} > p_{1}$, then $r_{i+1} = q_{i}+1(1< i < m $). 

For pattern \textbf{h}, $h_{1} = p_{1}+1$ and $h_{m+1} = p_{1}$. Moreover, if $q_{i} < p_{1}$, then $h_{i+1} = q_{i}$. If $q_{i} > p_{1}$, then $h_{i+1} = q_{i}+1(1< i < m $).

Example \ref {example4} illustrates the principle.

\begin{example}\label{example4}
	{Suppose there are only two frequent patterns with length-3, \textbf{p} = (2,1,3) and \textbf{q} = (1,3,2). Based on the two patterns, we show that different strategies can generate different number of candidate patterns with length-4. Table \ref{tab2} shows the sets of candidate patterns generated by enumeration and pattern fusion. If the enumeration method is adopted, there are four cases for each pattern, i.e., we can insert 1,2,3, or 4 at the end, while maintaining the relative order of the pattern (2,1,3). Thus, we get four candidate patterns (3,2,4,1), (3,1,4,2), (2,1,4,3), and (2,1,3,4), respectively. There are therefore eight candidate patterns using the enumeration strategy, since there are two length-3 patterns. }
	
	However, there are three candidate patterns using the pattern fusion strategy. We take (2,1,3)$\bigoplus$(1,3,2) as an example. Since $p_{1} = q_{3}$ = 2, according to Case 2, pattern fusion generates two candidate patterns, \textbf{r} and \textbf{h}. For pattern \textbf{r}, $r_{1} = p_{1}$ = 2 and $r_{4} = p_{1}+1$ = 3. Since $q_{1} = 1 < 2$, $r_{2} = q_{1}=1, $ and since $q_{2} = 3 > 2$, $r_{3} = q_{2}+1 = 4$. Hence, pattern \textbf{r} is (2,1,4,3). Similarly, pattern \textbf{h} is (3,1,4,2). Table \ref{tab2} shows a comparison of candidate patterns for these two different strategies. 
	
\begin{table}[!htb]
	\footnotesize
	\centering
	\caption{Comparison of candidate patterns }
	\label{tab2}
	\tabcolsep 4pt 
	\begin{tabular}{cccc}
		\hline\noalign{\smallskip}
		Frequent pattern & Enumeration  & Patterns & Pattern fusion  \\\hline
		(2,1,3) & \makecell*[c]{(3,2,4,1),(3,1,4,2)\\(2,1,4,3),(2,1,3,4)} & (2,1,3)$\bigoplus$(1,3,2) & (2,1,4,3),(3,1,4,2) \\
		(1,3,2) & \makecell*[c]{(2,4,3,1),(1,4,3,2)\\(1,4,2,3),(1,3,2,4)} & (1,3,2)$\bigoplus$(2,1,3) & (1,3,2,4) \\
		\noalign{\smallskip}\hline
	\end{tabular}
\end{table}
	
\end{example}

From Table \ref{tab2}, we can see that the pattern fusion strategy outperforms the enumeration strategy, since the pattern fusion strategy can prune many useless candidate patterns, thus improving the mining efficiency.

{Although the pattern fusion strategy was proposed in \cite{wu2021orde}, the correctness and completeness were not given in that paper. Now, we show the correctness and completeness as follows.}

{ \begin{theorem}\label {theorem0}
		Each candidate pattern is generated exact once and all candidate patterns can be generated, i.e., the pattern fusion strategy is correct and complete.
\end{theorem}}
\begin{proof}
	{Firstly, we show that the OPP mining satisfies the anti-monotonicity, which means that support of super-pattern \textbf{r} is less than that of its prefix pattern \textbf{p} or suffix pattern \textbf{q}. Suppose $<$$a$$>$ is an occurrence of super-pattern \textbf{r}. We can safely say that $<$$a-1$$>$ is an occurrence of pattern \textbf{p}, and $<$$a$$>$ is an occurrence of pattern \textbf{q}. Therefore, \textit{sup}(\textbf{r}, \textbf{t})$\leq$\textit{sup}(\textbf{p}, \textbf{t}) and  \textit{sup}(\textbf{r}, \textbf{t})$\leq$\textit{sup}(\textbf{q}, \textbf{t}). Hence, the OPP mining satisfies the anti-monotonicity. }

	{Secondly, we show that each candidate pattern can be generated only once. Proof by contradiction. Suppose  super-pattern \textbf{r} can be generated twice, and suppose  \textbf{r} is generated by two different prefix patterns. Suppose \textbf{r} = $(r_{1},r_{2},\ldots, r_{m},r_{m+1})$. Thus, its prefix pattern is $(r_{1},r_{2},\ldots, r_{m})$. According to Definition \ref{definition2}, we know that the relative order of $(r_{1},r_{2},\ldots, r_{m})$ is only one, i.e., the result of $R(r_{1},r_{2},\ldots, r_{m})$ is  an OPP, rather than two OPPs. This  contradicts the assumption that \textbf{r} is generated by two different prefix patterns. Hence, each candidate pattern is generated exact once.}
	
	{Finally, we show that all candidate patterns can be generated. Suppose  super-pattern \textbf{r} = $(r_{1},r_{2},\ldots, r_{m},r_{m+1})$ is not generated, the prefix and suffix patterns of \textbf{r} are \textbf{p} = $R(r_{1},r_{2},\ldots, r_{m})$ and \textbf{q} = $R(r_{2},\ldots, r_{m},r_{m+1})$, respectively. There are two cases: (1) pattern \textbf{p} or \textbf{q} is infrequent; (2) patterns \textbf{p} and \textbf{q} are frequent, but super-pattern \textbf{r} cannot be generated by the pattern fusion strategy.}
	
	{Case 1: Suppose pattern \textbf{p} is infrequent, i.e., \textit{sup}(\textbf{p}, \textbf{t})$<$$minsup$. Then, according to the anti-monotonicity,  \textit{sup}(\textbf{r}, \textbf{t})$<$$minsup$. Thus, pattern \textbf{r} is also infrequent. Hence, in this case, it is not necessary to generate super-pattern \textbf{r}. Similarly, if pattern \textbf{q} is infrequent, then it is not necessary to generate super-pattern \textbf{r}, either. }

	{Case 2: Proof by contradiction.  Suppose super-pattern \textbf{r} = $(r_{1},r_{2},\ldots, r_{m})$ cannot be generated by $\textbf{p}$ $\bigoplus$ $\textbf{q}$. We know that \textit{R}(\textit{suffix}(\textbf{p})) = \textit{R}(\textit{prefix}(\textbf{q})) = $R(r_{2},\ldots, r_{m})$. Therefore, we can generate  super-pattern \textbf{r} = $\textbf{p}$ $\bigoplus$ $\textbf{q}$ according to the pattern fusion strategy, which contradicts the assumption that super-pattern \textbf{r}    = $(r_{1},r_{2},\ldots, r_{m},r_{m+1})$ cannot be generated. Hence, all candidate patterns can be generated.}

\end{proof}

{For example, in Table \ref{tab2}, although patterns (3,2,4,1) and (2,1,3,4) cannot be generated by (2,1,3)$\bigoplus$(1,3,2), they can be generated by (2,1,3)$\bigoplus$(2,3,1) and (2,1,3)$\bigoplus$(1,2,3), respectively. This example shows that all patterns can be generated by using the pattern fusion strategy.}

\subsection{SPF for support calculation}\label {sub4.2}

{OPP-Miner adopts a pattern matching method to calculate pattern support, which does not use the calculation results of the sub-patterns \cite {wu2021orde}. If we can use the occurrences of subpatterns to generate the occurrences of super-patterns, then the new method can improve the efficiency, and is feasible.The reason is shown as follows. Suppose pattern \textbf{r} is generated by patterns \textbf{p} and \textbf{q}, i.e., \textbf{r}=\textbf{p} $\bigoplus$ \textbf{q}, and $<$$x$$>$ is an occurrence of pattern \textbf{r}. We can safely say that $<$$x-1$$>$ and $<$$x$$>$ are occurrences of patterns \textbf{p} and \textbf{q}, respectively. Similarly, we know that if $<$$x-1$$>$ is not an occurrence of pattern \textbf{p} or $<$$x$$>$ is not an occurrence of pattern \textbf{q}, then $<$$x$$>$ is not an occurrence of pattern \textbf{r}. An illustrative example is shown as follows.}

{For example, in Fig. \ref{fig:aa}, we know that the relative order of sub-time series ($t_2$,$t_3$,$t_4$,$t_5$) is (3,1,4,2), i.e., $<$$5$$>$ is an occurrence of pattern (3,1,4,2). Therefore, the relative orders of sub-time series ($t_2,t_3,t_4$) and ($t_3,t_4,t_5$) are (2,1,3) and (1,3,2), respectively. Moreover, the relative orders of sub-time series ($t_{12}$,$t_{13}$,$t_{14}$) is (2,1,3), but that of ($t_{13}$,$t_{14}$,$t_{15}$) is not (1,3,2). Therefore, $<$$15$$>$ is not an occurrence of pattern (3,1,4,2). Hence, we propose an algorithm called SPF to calculate the support based on pattern fusion, which can use the occurrences of sub-patterns to generate  the occurrences of super-patterns. The details are shown as follows.}

{From Section \ref{sub4.1},  super-patterns \textbf{r} and \textbf{h} are generated by  \textbf{p} $\oplus$ \textbf{q} which can be seen as the prefix and suffix patterns of the super-patterns, respectively. Suppose $<$$lp_{i}$$>$ and $<$$lq_{j}$$>$ are the occurrences of \textbf{p} and \textbf{q}, respectively. All occurrences of \textbf{p} and \textbf{q} are stored in a prefix array $\mathcal{P}_{\textbf{p}}$ and a suffix array $\mathcal{S}_{\textbf{q}}$, respectively, i.e., $<$$lp_{i}$$> \in \mathcal{P}_{\textbf{p}}$ and $<lq_{j}> \in \mathcal{S}_{\textbf{q}}$. The matching results of super-patterns \textbf{r} and \textbf{h} are stored in $L_{\textbf{r}}$ and $L_{\textbf{h}}$, respectively. This method is demonstrated as follows.}

\textbf{Rule 1.} If $p_{1}\ne q_{m}$, then \textbf{r} = \textbf{p} $\oplus$ \textbf{q}:

{As shown in Fig. \ref{fig:bb}, if and only if $lq_{j}=lp_{i}+1$, then $<$$lq_{j}$$>$ is an occurrence of \textbf{r}, i.e., $lq_{j}$$\in$ $L_{\textbf{r}}$.}

\begin{figure}[!htb]
	\centering
	\includegraphics[width=\linewidth]{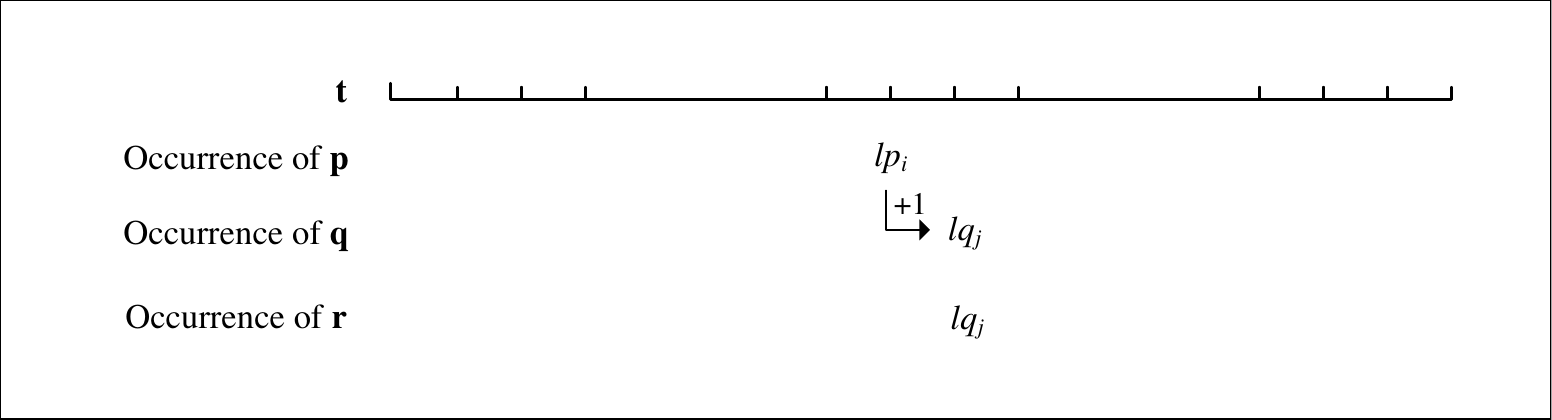}
	\caption{Occurrence of \textbf{r} in sequence \textbf{t}}
	\label{fig:bb}		
\end{figure}

\textbf{Rule 2.} If $p_{1} = q_{m}$, then \textbf{r}, \textbf{h}= \textbf{p} $\oplus$ \textbf{q}:

As shown in Figure \ref{fig:cc}, if $lq_{j} = lp_{i}+1$, then $<$$lq_{j}$$>$ may be an occurrence of \textbf{r} or \textbf{h}. It is necessary to determine $t_{begin}$ and $t_{end}$ in t, where \textit{begin} = $lq_{j} - m$ and \textit{end} = $lq_{j}$. There are three cases:

{\textbf{Case 1:} If $t_{begin}<t_{end}$, then $<$$lq_{j}$$>$ is an occurrence of \textbf{r}, i.e., $lq_{j}$$\in$ $L_{\textbf{r}}$. }

{\textbf{Case 2:} If $t_{begin}>t_{end}$, then $<$$lq_{j}$$>$ is an occurrence of \textbf{h}, i.e., $lq_{j}$$\in$ $L_{\textbf{h}}$. }

{\textbf{Case 3: }If $t_{begin}=t_{end}$, then $<$$lq_{j}$$>$ is an occurrence of neither \textbf{r} nor \textbf{h}.}

\begin{figure}[!htb]
	\centering
	\includegraphics[width=\linewidth]{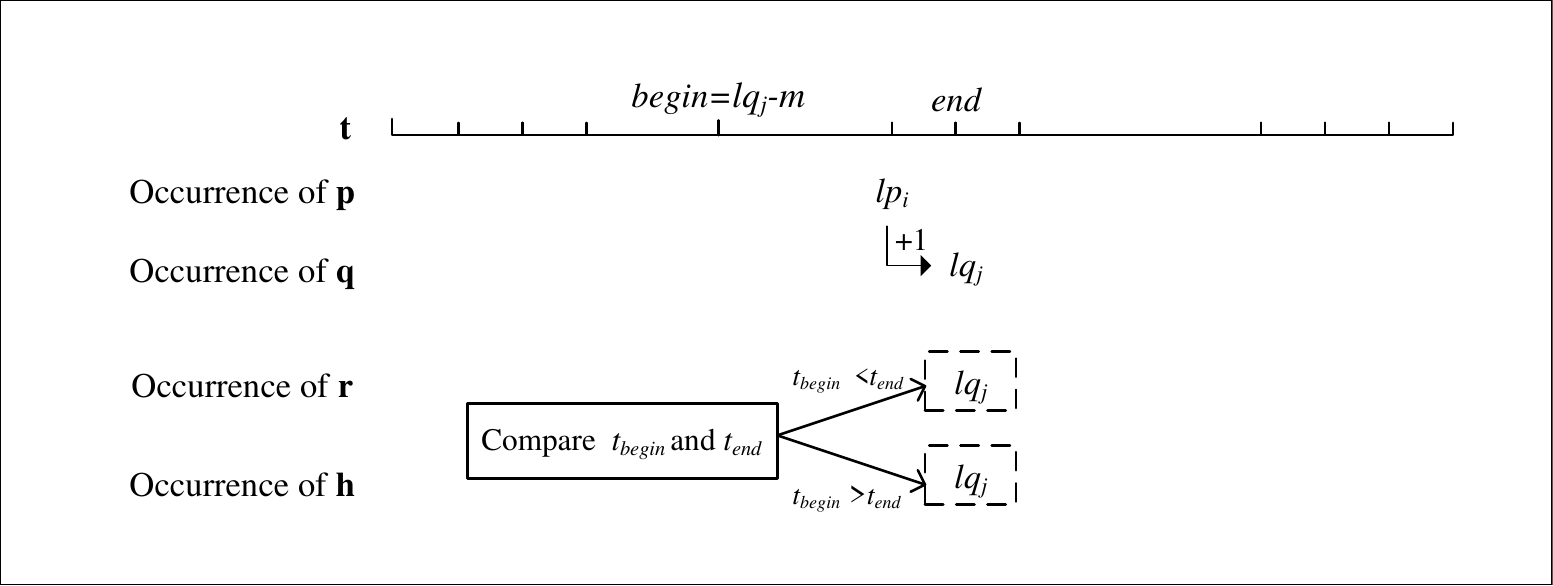}
	\caption{Occurrence of \textbf{r} and \textbf{h} in sequence \textbf{t}}
	\label{fig:cc}		
\end{figure}

Finally, the size of sets $L_{\textbf{r}}$ and $L_{\textbf{h}}$ are the supports of the super-patterns \textbf{r} and \textbf{h}, respectively, i.e., \textit{sup}(\textbf{r}) = $|L_{\textbf{r}}|$ and \textit{sup}(\textbf{h}) = $|L_{\textbf{h}}|$. An illustration is given in Example 5.

\begin{example}\label{example5}
	{	Suppose we have a time series \textbf{t}, as shown in Table \ref{tab3}. The matching sets of length-2 patterns \textbf{p} = (1,2) and \textbf{q} = (2,1) are $L_{\textbf{p}} = \{$$<$2$>$,$<$4$>$,$<$8$>$,$<$10$>$,$<$12$>$,$<$14$>$,$<$16$>$$\}$ and $L_{\textbf{q}} = \{$$<$3$>$,$<$5$>$,$<$6$>$,$<$7$>$,$<$9$>$,$<$11$>$,$<$13$>$,$<$15$>$$\}$, respectively.}
\end{example}

\begin{table}[!htb]
	\centering
	\setlength{\belowcaptionskip}{0.1cm}
	\caption{{Time series element index }}
	\label{tab3}
	\tabcolsep 3.5pt 
	\begin{tabular}{ccccccccccccccccc}
		\hline\noalign{\smallskip}
		\textit{ID}&1&2&3&4&5&6&7&8&9&10&11&12&13&14&15&16  \\
		\hline
		\textit{t}&24&31&27&33&30&24&21&25&23&26&22&27&24&28&23&29  \\
		\hline\noalign{\smallskip}
	\end{tabular}
\end{table}

{\textbf{p} $\oplus$ \textbf{q} can generate two super-patterns, \textbf{r} = (1,3,2) and \textbf{h} = (2,3,1). We know that $\mathcal{P}_{\textbf{p}}$ = $L_{\textbf{p}}$, since \textbf{p} is a prefix pattern. Similarly, $\mathcal{S}_{\textbf{q}}$ = $L_{\textbf{q}}$. Moreover, 2 $\in \mathcal{P}_{\textbf{p}}$ and 2+1 = 3 $\in\mathcal{S}_{\textbf{q}}$. Hence, according to Rule 2, $<$$3$$>$ may be an occurrence of \textbf{r} or \textbf{h}.  \textit{begin} = 3$-$2 = 1 and \textit{end} = 3. Thus, $<$$3$$>$ is one occurrence of \textbf{r}, that is, $<$3$>$ $\in$ $L_{\textbf{r}}$, since $t_{1} = 24 < t_{3}$ = 27. Similarly, we know that the matching set of \textbf{r} is $L_{\textbf{r}} = \{$$<$3$>$,$<$5$>$,$<$9$>$,$<$11$>$,$<$13$>$,$<$15$>$$\}$ and \textit{sup}(\textbf{r}) = 6. The matching set of \textbf{h} is $L_{\textbf{h}} = $\{$<$11$>$,$<$15$>$$\}$  and \textit{sup}(\textbf{h}) = 2.}

\subsection{SPF-Pro for support calculation}\label {sub4.22}

In the SPF algorithm, $\mathcal{P}_{\textbf{p}}$ and $\mathcal{S}_{\textbf{q}}$ are fixed. To further improve the efficiency of SPF, we propose a more efficient approach called SPF-Pro, in which $\mathcal{P}_{\textbf{p}}$ and $\mathcal{S}_{\textbf{q}}$ are reduced dynamically, where the initial values of $\mathcal{P}_{\textbf{p}}$ and $\mathcal{S}_{\textbf{q}}$ are $L_{\textbf{p}}$ and $L_{\textbf{q}}$, respectively, i.e., $\mathcal{P}_{\textbf{p}}$ = $L_{\textbf{p}}$ and $\mathcal{S}_{\textbf{q}}$ = $L_{\textbf{q}}$. This method is called a screening strategy.

\textbf{Screening strategy.} In Rules 1 and 2, if $lp_{i}$ in $\mathcal{P}_{\textbf{p}}$ and $lq_{j}$ in $\mathcal{S}_{\textbf{q}}$ are used to generate an occurrence of \textbf{r} or \textbf{h}, then $lp_{i}$ and $lq_{j}$ can be pruned. The new corresponding rules are shown below as Rules 3 and 4, respectively.

\textbf{Rule 3.} If $p_{1}\ne q_{m}$, then \textbf{r} = \textbf{p} $\oplus$ \textbf{q}:

{If and only if $lq_{j}=lp_{i}+1$, then $<$$lq_{j}$$>$ is an occurrence of \textbf{r}, i.e., $lq_{j}$$\in$ $L_{\textbf{r}}$, and $lp_{i}$ and $lq_{j}$ are pruned from $\mathcal{P}_{\textbf{p}}$ and $\mathcal{S}_{\textbf{q}}$, respectively.}

\textbf{Rule 4.} If $p_{1} = q_{m}$, then \textbf{r}, \textbf{h} = \textbf{p} $\oplus$ \textbf{q}:

{If and only if $lq_{j}=lp_{i}+1$, $<$$lq_{j}$$>$ may be an occurrence of \textbf{r} or \textbf{h}. There are then three cases: }

{\textbf{Case 1}: If $t_{begin}<t_{end}$, then $<$$lq_{j}$$>$ is an occurrence of \textbf{r}, i.e., $lq_{j}$$\in$ $L_{\textbf{r}}$, and $lp_{i}$ and $lq_{j}$ are pruned from $\mathcal{P}_{\textbf{p}}$ and $\mathcal{S}_{\textbf{q}}$, respectively. }

{\textbf{Case 2}: If $t_{begin}>t_{end}$, then $<$$lq_{j}$$>$ is an occurrence of \textbf{h}, i.e., $lq_{j}$$\in$ $L_{\textbf{h}}$, and $lp_{i}$ and $lq_{j}$ are pruned from $\mathcal{P}_{\textbf{p}}$ and $\mathcal{S}_{\textbf{q}}$, respectively. }

{\textbf{Case 3}: If $t_{begin}=t_{end}$, then $<$$lq_{j}$$>$ is an occurrence of neither \textbf{r} nor \textbf{h}.}

To prove the correctness of this screening strategy, we initially prove two theorems.

\begin{theorem}\label {theorem1}
	Suppose \textbf{p} can fuse with $\textbf q_{1}$ and $\textbf q_{2}$, i.e.,  \textbf{r}$_{1}$, \textbf{h}$_{1}$ = \textbf{p} $\bigoplus$ \textbf{q}$_{1}$ and \textbf{r}$_{2}$, \textbf{h}$_{2}$ = \textbf{p} $\bigoplus$ \textbf{q}$_{2}$. If $lp_{i}$ + 1 = $x\in L_{\textbf{r}_1}$ or $L_{\textbf{h}_1}$, then $x\notin L_{\textbf{r}_2}$ or $L_{\textbf{h}_2}$, and vice versa.
\end{theorem}
\begin{proof}
	{	 (Proof by contradiction) Suppose  $lp_{i}$ + 1 = $x\in L_{\textbf{r}_1}$ and $x\in L_{\textbf{r}_2}$. Since $x\in L_{\textbf{r}_1}$, we know that $<$$x$$>$ is an occurrence of \textbf{r}$_{1}$. Similarly, $<$$x$$>$ is also an occurrence of \textbf{r}$_{2}$. Obviously, $<$$x$$>$ cannot be two occurrences for two different patterns with the same length. Hence, this does not hold and the assumption is contradicted; that is, $x\notin L_{\textbf{r}_2}$, and vice versa.} 
\end{proof}
\begin{theorem}\label{theorem2}
	Suppose \textbf{p$_{1}$} and \textbf{p$_{2}$} can fuse with \textbf{q}, i.e., \textbf{r}$_{1}$, \textbf{h}$_{1}$ = \textbf{p$_{1}$} $\bigoplus$ \textbf{q} and \textbf{r}$_{2}$, \textbf{h}$_{2}$ = \textbf{p}$_{2}$ $\bigoplus$ \textbf{q}. If $lq_{j}$ = $x\in L_{\textbf{r}_1}$ or $L_{\textbf{h}_1}$, then $x\notin L_{\textbf{r}_2}$ or $L_{\textbf{h}_2}$, and vice versa.
\end{theorem}
\begin{proof}
	The proof method is the same as for Theorem \ref{theorem1}.
\end{proof}
\begin{theorem}\label{theorem4}
	The screening strategy is correct.
\end{theorem}
\begin{proof}
	{According to Theorem \ref{theorem1}, $<$$lp_{i}$$>$ belongs to only one pattern. Hence, if $<$$lp_{i}$$>$ is used to generate an occurrence of its super-pattern, then $<$$lp_{i}$$>$ can be pruned. Similarly, according to Theorem \ref {theorem2}, $<$$lq_{j}$$>$ can also be pruned. We have therefore proved the correctness of the screening strategy.}
\end{proof}

Example \ref{example6} is used to demonstrate that SPF-Pro outperforms SPF.

\begin{example}\label{example6}
	{	We adopt the same data as in Example \ref {example5}. We know that \textbf{q}=(2,1), and \textbf{q} can fuse with \textbf{q}, i.e., \textbf{e} = \textbf{q}$\bigoplus$\textbf{q} = (3,2,1). $\mathcal{P}_{\textbf{q}}$ = $\mathcal{S}_{\textbf{q}}$ = $L_{\textbf{q}}$ = $\{$$<$3$>$, $<$5$>$,$<$6$>$,$<$7$>$,$<$9$>$,$<$11$>$,$<$13$>$,$<$15$>$$\}$.   According to SPF, we know that  $L_{\textbf{e}}$ = $\{$$<$6$>$,$<$7$>$$\}$ and sup(\textbf{e}) = 2.}
	
	{We now show that SPF-Pro yields better performance than SPF. In Example \ref{example5}, we know that the super-patterns \textbf{r} and \textbf{h} are generated. According to the screening strategy, $<$$3$$>$ is an occurrence of \textbf{r} = (1,3,2). Hence, $3\notin\mathcal{S}_{\textbf{q}}$, and 3 is pruned from $\mathcal{S}_{\textbf{q}}$. Similarly, according to SPF-Pro, we know that $\mathcal{S}_{\textbf{q}}$ = $\{$$<$6$>$,$<$7$>$$\}$. SPF-Pro then uses $\mathcal{P}_{\textbf{q}}$ = $\{$$<$3$>$,$<$5$>$,$<$6$>$,$<$7$>$,$<$9$>$,$<$11$>$,$<$13$>$,$<$15$>$$\}$ and $\mathcal{S}_{\textbf{q}}$ = $\{$$<$6$>$,$<$7$>$$\}$ to calculate the support of \textbf{e}. Moreover, $L_{\textbf{e}}$ = $\{$$<$6$>$,$<$7$>$$\}$ and \textit{sup}(\textbf{e}) = 2, which are the same as for SPF. Now, we can see that in SPF, $\mathcal{S}_{\textbf{q}}$ = $L_{\textbf{q}}$ = $\{$$<$3$>$, $<$5$>$,$<$6$>$,$<$7$>$,$<$9$>$,$<$11$>$,$<$13$>$,$<$15$>$$\}$, while in SPF-Pro, $\mathcal{S}_{\textbf{q}}$ = $\{$$<$6$>$,$<$7$>$$\}$, with a size that is significantly smaller than in SPF. Hence, SPF-Pro outperforms SPF.}
\end{example}

Pseudocode for SPF-Pro is given in Algorithm \ref {Algorithm 2}, which calculates the supports of the super-patterns using the pattern fusion strategy.

\begin{algorithm}[htb]
	\caption{SPF-Pro}	\label{Algorithm 2}
	{\bf Input:}
	Pattern \textbf{p} and its matching result $\mathcal{P}_{\textbf{p}}$, pattern \textbf{q} and its matching result $\mathcal{S}_{\textbf{q}}$
	{\bf Output:} Super-patterns and their matching results, and $\mathcal{P}_{\textbf{p}}$ and $\mathcal{S}_{\textbf{q}}$
	\begin{algorithmic}[1]
		\State $L_{\textbf{r}}$ = $\{\}$;$L_{\textbf{h}}$ = $\{\}$;
		\State \textbf{e} $\leftarrow$ \textit{R}(\textit{prefix}(\textbf{p}));  
		\State \textbf{k} $\leftarrow$ \textit{R}(\textit{suffix}(\textbf{q}));
		\If {\textbf{k} == \textbf{e} }	
		\If {\textbf{p}[0] == \textbf{q}[$m-$1] }
		\State \textbf{r}$\cup$\textbf{h}  $\leftarrow$ \textbf{p}$\oplus$\textbf{q};
		\State Calculate $L_{\textbf{r}}$ and $L_{\textbf{h}}$, and update $\mathcal{P}_{\textbf{p}}$ and $\mathcal{S}_{\textbf{q}}$ according to Rule 4;
		\Else
		\State \textbf{r} $\leftarrow$ \textbf{p}$\oplus$\textbf{q};
		\State Calculate $L_{\textbf{r}}$, and update $\mathcal{P}_{\textbf{p}}$ and $\mathcal{S}_{\textbf{q}}$ according to Rule 3;
		\EndIf
		\EndIf
		\State Return \textbf{r}, \textbf{h}, $\mathcal{P}_{\textbf{p}}$ and $\mathcal{S}_{\textbf{q}}$;
	\end{algorithmic}
\end{algorithm}
\subsection{Pruning candidate patterns}

\label {sub4.3}
In this section, we propose a pruning strategy to further prune candidate patterns based on SPF-Pro.

\textbf{Pruning strategy.} If $|\mathcal{P}_{\textbf{p}}|$ $<$ \textit{minsup}, then \textbf{p} as a prefix pattern will no longer generate frequent patterns. If $|\mathcal{S}_{\textbf{q}}|$ $<$ \textit{minsup}, then \textbf{q} as a suffix pattern will no longer generate frequent patterns.
\begin{theorem}\label{theorempruning}
	The pruning strategy is correct.
\end{theorem}

\begin{proof}
	Suppose pattern \textbf{p} can fuse with pattern \textbf{q}, i.e., \textbf{r}, \textbf{h} = \textbf{p}$\oplus$\textbf{q}. Obviously, the sizes of $L_{\textbf{r}}$ and $L_{\textbf{h}}$ are not greater than the size of $\mathcal{P}_{\textbf{p}}$ or $\mathcal{S}_{\textbf{q}}$, since according to SPF-Pro, if and only if $lq_{j}=lp_{i}+1$ ($lp_{i}\in \mathcal{P}_{\textbf{p}}$, $lq_{j}\in \mathcal{S}_{\textbf{q}}$), $lq_{j}\in L_{\textbf{r}}$. Thus, $|L_{\textbf{r}}| \leq |\mathcal{P}_{\textbf{p}}|$. Therefore, $|L_{\textbf{r}}| <$ \textit{minsup}, since $|\mathcal{P}_{\textbf{p}}| <$  \textit{minsup}. Hence, \textbf{p} as a prefix pattern will no longer generate frequent patterns. Similarly, we can prove that \textbf{q} as a suffix pattern will no longer generate frequent patterns.
\end{proof}

Example \ref{example7} illustrates the effectiveness of  pruning strategy.
\begin{example}\label{example7}
	We use the same data as in Example \ref{example6}. We know that \textbf{p} = (1,2) and \textbf{q} = (2,1). According to Rule 4, after  two patterns \textbf{r} = (1,3,2) and \textbf{h} = (2,3,1) are  generated by \textbf{p}$\bigoplus$\textbf{q}, we know that $\mathcal{S}_{\textbf{q}}$ = $\{$$<$6$>$,$<$7$>$$\}$. Suppose \textit{minsup} = 3. If we do not apply the pruning strategy, according to Example 6, we have to use $\mathcal{P}_{\textbf{q}}$ = $\{$$<$3$>$, $<$5$>$,$<$6$>$,$<$7$>$,$<$9$>$,$<$11$>$,$<$13$>$,$<$15$>$$\}$ and $\mathcal{S}_{\textbf{q}}$ = $\{$$<$6$>$,$<$7$>$$\}$ to calculate the support of \textbf{e} = \textbf{q}$\bigoplus$\textbf{q} = (3,2,1). We know that $L_{\textbf{e}}$ = $\{$$<$6$>$,$<$7$>$$\}$ and \textit{sup}(\textbf{e}) = 2, and pattern \textbf{e} is not a frequent pattern. However, according to the pruning strategy, we do not need to use $\mathcal{P}_{\textbf{q}}$ and $\mathcal{S}_{\textbf{q}}$ = $\{$$<$6$>$,$<$7$>$$\}$ to calculate the support of e, since $|\mathcal{S}_{\textbf{q}}|$ = 2 $<$ \textit{minsup}. Hence, we can avoid calculating \textbf{q}$\bigoplus$\textbf{q} using this approach.
\end{example}

\subsection{Mining OPPs}\label{sub4.4}

In this section, we introduce the EFO-Miner algorithm to discover frequent OPPs.

The steps of EFO-Miner are as follows.

Step 1: Scan the time series \textbf{t} to calculate the matching results and the supports of patterns (1,2) and (2,1). If the pattern is frequent, then it is stored into the frequent pattern set $F_{2}$;

Step 2: Select any two patterns \textbf{p} and \textbf{q} in $F_{m}$. If pattern \textbf{p} can fuse with pattern \textbf{q}, then \textbf{p}$\bigoplus$\textbf{q} can generate candidate super-patterns \textbf{r} and \textbf{h}. If $|\mathcal{P}_{\textbf{p}}|\geq minsup$ and $|\mathcal{S}_{\textbf{q}}|\geq minsup$, then use SPF-Pro to calculate the matching results and the supports of super-patterns \textbf{r} and \textbf{h}. If \textbf{r} or \textbf{h} is frequent, store it in the set $F_{m+1}$;

Step 3: Iterate Step 2 until no (\textit{m}+1)-length super-pattern is generated;

Step 4: Iterate Steps 2 and 3 until $F_{m+1}$ is empty. 

Finally, all frequent patterns $F$ = $F_{2}\cup F_{3}\cup...F_{m}$.

Example 8 illustrates the principle of EFO-Miner.

\begin{example}\label{example8}
	We use the same data as in Example \ref{example5}. Suppose \textit{minsup} = 3. We can discover all frequent patterns as follows.
	
	First, the matching sets of length-2 patterns \textbf{p} = (1,2) and \textbf{q} = (2,1) are $L_{\textbf{p}}$ = $\{$$<$2$>$,$<$4$>$,$<$8$>$,$<$10$>$, $<$12$>$,$<$14$>$,$<$16$>$$\}$ and $L_{\textbf{q}}$ = $\{$$<$3$>$,$<$5$>$,$<$6$>$,$<$7$>$,$<$9$>$, $<$11$>$,$<$13$>$,$<$15$>$$\}$, respectively. Therefore, $\mathcal{P}_{\textbf{p}}$ = $\mathcal{S}_{\textbf{p}}$ = $L_{\textbf{p}}$ = $\{$$<$2$>$,$<$4$>$,$<$8$>$,$<$10$>$, $<$12$>$,$<$14$>$,$<$16$>$$\}$ and $\mathcal{P}_{\textbf{q}}$ = $\mathcal{S}_{\textbf{q}}$ = $L_{\textbf{q}}$ = $\{$$<$3$>$,$<$5$>$,$<$6$>$,$<$7$>$,$<$9$>$,$<$11$>$,$<$13$>$, $<$15$>$$\}$. Since \textit{sup}(\textbf{p}) = 7 and \textit{sup}(\textbf{q}) = 8, we know that $F_{2}$ = $\{$(1,2), (2,1)$\}$.
	
	EFO-Miner now finds frequent patterns with length three. \textbf{p}$\bigoplus$\textbf{p} = (1,2)$\bigoplus$(1,2) = (1,2,3). According to SPF-Pro, \textit{sup}(1,2,3) = 0, and $\mathcal{P}_{\textbf{q}}$ and $\mathcal{S}_{\textbf{q}}$ are not changed. Similarly, \textbf{p}$\bigoplus$\textbf{q} generates two candidate patterns, (1,3,2) and (2,3,1). SPF-Pro calculates $L_{(1,3,2)}$ = $\{$$<$3$>$, $<$5$>$,$<$9$>$,$<$11$>$,$<$13$>$,$<$15$>$$\}$, \textit{sup}(1,3,2) = 6 and \textit{sup}(2,3,1) = 0. Meanwhile, $\mathcal{P}_{\textbf{p}}$ = $\{$16$\}$ and $\mathcal{S}_{\textbf{q}}$ = $\{$$<$6$>$,$<$7$>$$\}$. Thus, (1,3,2) is a frequent pattern. When the pruning strategy is used, \textbf{p} as a prefix pattern and \textbf{q} as a suffix pattern will no longer generate frequent patterns. In a similar way, (2,1,3) can be found. Hence, the length-3 frequent pattern set $F_{3}$ = $\{$(1,3,2), (2,1,3)$\}$ is obtained. Moreover, length-4 frequent patterns can be calculated based on $F_{3}$. Finally, we get the frequent pattern set \textit{F} = $\{$(1,2), (2,1), (1,3,2), (2,1,3), (1,3,2,4), (3,1,4,2)$\}$.
	
\end{example}

Pseudocode for EFO-Miner is given in Algorithm \ref {Algorithm 3}.

\begin{algorithm}[htb]
	
	\caption{EFO-Miner}	\label{Algorithm 3}
	{\bf Input:} Time series \textbf{t} and the minimum support threshold \textit{minsup}
	{\bf Output:} Frequent pattern set \textit{F}
	\begin{algorithmic}[1]		
		\State Scan sequence \textbf{t}, and use $L_{(1,2)}$ and $L_{(2,1)}$ to store the matching sets of (1,2) and (2,1), respectively. If the size of $L_{(1,2)}$ is no less than \textit{minsup}, then add pattern (1,2) to $F_{2}$. Similarly, add pattern (2,1) to $F_{2}$;
		\State \textit{m} $\leftarrow$ 2;  
		\While {$F_{m}$ $<>$ NULL}
		\For {each \textbf{p} in $F_{m}$}
		\State $\mathcal{P}_{\textbf{p}}$ $\leftarrow$ $\mathcal{S}_{\textbf{p}}$ $\leftarrow$ $L_{\textbf{p}}$; 
		\EndFor
		\For {each \textbf{p} in $F_{m}$}
		\For {each \textbf{q} in $F_{m}$}
		\If {$\mathcal{P}_{\textbf{p}}.size\geq minsup$ $\&\&$ $\mathcal{S}_{\textbf{q}}.size\geq minsup$}
		\If {\textbf{p} can fuse with \textbf{q}}
		\State Calculate the matching results of super-patterns \textbf{r} and \textbf{h} and update $\mathcal{P}_{\textbf{p}}$ and $\mathcal{S}_{\textbf{q}}$ using SPF-Pro;
		\If {$L_{\textbf{r}}.size\geq minsup$} 
		\State Add pattern \textbf{r} to $F_{m+1}$;
		\EndIf
		\If {$L_{\textbf{h}}.size\geq minsup$ }
		\State Add pattern \textbf{h} to $F_{m+1}$;
		\EndIf
		\EndIf
		\EndIf
		\EndFor
		\EndFor
		\State \textit{m} $\leftarrow$ \textit{m}+1
		\EndWhile
		\State Return \textit{F};
	\end{algorithmic}
\end{algorithm}

\begin{theorem} \label{theoremefo}
	{EFO-Miner is correct and complete.}
\end{theorem}
\begin{proof}
	{We know that EFO-Miner employs the pattern fusion strategy to generate candidate patterns, the screening strategy to calculate the supports of candidate patterns, and the pruning strategy to further prune candidate patterns. Theorems \ref{theorem0}, \ref{theorem4}, and \ref{theorempruning} show the correctness and completeness of these strategies. Therefore, EFO-Miner is correct and complete.}
\end{proof}

\begin{theorem} \label{theorem5}
	The space and time complexity of EFO-Miner are $O(f \times n)$, where $f$ and $n$ are the number of frequent patterns and the sequence length, respectively.
\end{theorem}
\begin{proof}
	The space complexity of EFO-Miner involves two parts: the frequent patterns and the matching results. Since the number of frequent patterns is f, the space complexity of frequent patterns is $O(f \times m)$, where m is the length of the longest pattern. For each pattern \textbf{p}, the space complexity of the matching results is $O(n)$. Similarly, the space complexities of $\mathcal{P}_{\textbf{p}}$ and $\mathcal{S}_{\textbf{p}}$ are also O(n). Since there are $f$ frequent patterns, the space complexity of the matching results is $O(f \times n)$. Since $m$ is far less than $n$, the space complexity of EFO-Miner is $O(f \times (m+n))$ = $O(f \times n)$. 	The time complexity of calculating the matching results for each pattern is $O(n)$. There are $f$ patterns. Therefore, the time complexity of EFO-Miner is $O(f \times n)$.
\end{proof}

\subsection{Mining strong OPRs }\label{sub4.5}
In this section, we explore the use of OPR-Miner to mine strong OPRs from all frequent patterns using EFO-Miner.

A simple method is that we  enumerate all OPRs according to Definition \ref{definition6} and calculate their confidences. If the confidence is no less than the threshold, then the rule is a strong OPR. Obviously, this method is not efficient. 

According to Algorithm 	\ref{Algorithm 3}, we know that pattern {\textbf{p}} is the prefix OPP of patterns {\textbf{r}} and {\textbf{h}}. Therefore, we can discover the strong OPRs in the process of mining frequent OPPs. It means that if the support of pattern {\textbf{r}} is no less than $minsup/minconf$, then $\textbf{p} \to \textbf{r}$ is a strong OPR. Similarly, $\textbf{p} \to \textbf{h}$ is a strong OPR. More importantly, this method has the same time and space complexities as those of EFO-Miner. Pseudocode for OPR-Miner is shown in Algorithm \ref{Algorithm 4}.

\begin{algorithm}[htb]
	\caption{OPR-Miner}	\label{Algorithm 4}
	
	{\bf Input:} Time series \textbf{t}, frequent pattern set \textit{F}, support of each frequent pattern \textit{sup}, and the minimum confidence threshold \textit{minconf}
	{\bf Output:} Strong OPR set \textit{R}
	
	\begin{algorithmic}[1]
		
		\If {$L_{\textbf{r}}.size/L_{\textbf{p}}.size\geq minconf$}  // Add these codes after Line 17 in Algorithm 	\ref{Algorithm 3}.
		\State Add rule $\textbf{p} \to \textbf{r}$ to \textit{R};
		\EndIf
		\If {$L_{\textbf{h}}.size/L_{\textbf{p}}.size\geq minconf $ }
		\State Add rule $\textbf{p} \to \textbf{r}$ to \textit{R};
		\EndIf
		
	\end{algorithmic}
\end{algorithm}

{According to Algorithm \ref{Algorithm 4}, we know that OPR-Miner does not employ any strategy, only uses Definitions \ref{definition7} and \ref{definition8} to discover strong OPRs based on EFO-Miner. Theorem \ref{theoremefo} shows that EFO-Miner is correct and complete. Therefore, OPR-Miner is also correct and complete. }

{Moreover, Example 9 illustrates the difference between all OPRs and strong OPRs.}

\begin{example}\label{example9}
	{	This example uses the frequent OPPs in Example \ref{example8}. We know that (1,2) and (1,3,2) are two frequent patterns, where (1,2) is the prefix pattern of (1,3,2). According to Definition \ref{definition6}, (1,2)$\to$(1,3,2) is an OPR. Similarly, we find all OPRs: (1,2)$\to$(1,3,2), (2,1)$\to$(2,1,3), (1,3,2)$\to$(1,3,2,4), and (2,1,3)$\to$(3,1,4,2).}
	
	{However, according to Definition \ref{definition7}, the confidence of rule (2,1)$\to$(2,1,3) is \textit{conf}((2,1)$\to$(2,1,3)) = \textit{sup}(2,1,3)$ $/$ $\textit{sup}(2,1) = 4$/$8 = 0.5. Since rules with low confidence have no practical meaning in most applications, we only discover the strong OPRs, that is, those for which the confidence level is higher than a certain threshold. For example, suppose the minimum confidence threshold minconf is 0.7. Thus, rule (2,1)$\to$(2,1,3) is not a strong OPR, since its confidence is 0.5. We know that \textit{sup}((1,2)) = 7 and \textit{sup}(1,3,2) = 6. Hence, the confidence of rule (1,2)$\to$(1,3,2) is \textit{conf}((1,2)$\to$(1,3,2)) = 6/7, which is greater than \textit{minconf}, and rule (1,2)$\to$(1,3,2) is therefore a strong OPR. Similarly, we get the strong OPR set \textit{R} = \{(1,2)$\to$(1,3,2), (2,1,3)$\to$(3,1,4,2)\}.}
\end{example}
{This example shows that the number of strong OPRs is less than that of all OPRs, since the confidence of a strong OPR is no less than \textit{minconf}, while the general OPRs do not have such constraints.   }

\section{Experimental results and analysis}\label{section5}

Section \ref{sub5.1} introduces the benchmark datasets and the baseline methods.  Section \ref{sub5.2} validates the running  performance of EFO-Miner. Section \ref{scalabilty} shows the scalabilty of EFO-Miner. Section \ref{minsup} reports the influence of different \textit{minsup}. Section \ref{sub5.3} verifies the performance of OPR-Miner. Section \ref{subconf} shows the the influence of different \textit{minconf}. Section \ref{sub5.4} demonstrates the advantages of OPR patterns.

\subsection{Benchmark datasets and baseline methods}\label{sub5.1}

We use real stock, weather, and oil datasets as test sequences. The stock and oil datasets can be downloaded from https://www.yahoo.com/, the weather datasets can be downloaded from https://archive.ics.uci.edu/ml/datasets.php/, the daily new cases datasets can be downloaded from https://coronavirus.jhu.edu/, the sensor and spectro datasets can be downloaded from http://www.timeseriesclassification.com/index.php/, and the diagnosis fault datasets can be  downloaded from http://jzw.ie.tsinghua.edu.cn/Show/index/cid/45/id/1568.html/. A specific description of each dataset is given in Table \ref{tab4}.

\begin{table}[!htb]
	\centering
	\setlength{\belowcaptionskip}{0.1cm}
	\caption{Description of datasets}
	\label{tab4}
	\tabcolsep 1pt 
	\begin{tabular}{cccccc}
		\hline\noalign{\smallskip}
		Name & Dataset & Type & Total length &  Number of sequences &  Number of labels\\\hline
		SDB1 & Italian-temperature & Weather & 256 & 1 & /\\
		SDB2 & Italian-temperature & Weather & 1,233 & 1 & /\\
		SDB3 & 1WTl-2 & Oil & 2,496 & 1 & /\\
		SDB4 & Crude Oil & Oil & 4,954 & 1 & /\\
		SDB5 & Russell 2000 & Stock & 8,141 & 1 & /\\
		SDB6 & Nasdaq & Stock & 12,279 & 1 & /\\
		SDB7 & S\&P 500 & Stock & 23,046 & 1 & /\\
		SDB8 & PRSA\_Data\_Nongzhanguan & Weather & 34,436 & 1 & /\\
		SDB9 & CSSE COVID\-19 Dataset & Daily new cases & 2,715 & 15 &15\\
		SDB10 & Car & Sensor & 8,655 & 15 & 4\\
		SDB11 & Meat & Spectro & 6,345 & 15 & 3\\
		SDB12 & Beef & Spectro & 7,050 & 15 & 5\\
		SDB13 & Bearing fault-NR & Diagnosis fault & 46,024 & 44 & 2\\
		SDB14 & Bearing fault-NI & Diagnosis fault & 46,024 & 44 & 2\\
		SDB15 & Bearing fault-NO & Diagnosis fault & 46,024 & 44 & 2\\	
		SDB16  &  New York Stock Exchange   &  Stock  &  60,000  & 1 & /\\
		\hline\noalign{\smallskip}
	\end{tabular}
	\begin{tablenotes}
		\item Note: SDB13-SDB15 are part of the sequences selected from the bearing fault dataset, which records the bearing fault vibration signals. There are four bearing State labels representing different States. Normal, Inner, Outer, and Roller. Among them, 22 Normal and 22 Roller sequences are extracted from SDB13, 22 Normal and 22 Inner sequences are extracted from SDB14, and 22 Normal and 22 outer sequences are extracted from SDB15. 
	\end{tablenotes}
\end{table}

All experiments were run on a computer with Intel(R) Core(TM) i5-3230U, 1.60 GHz CPU, 8.0 GB RAM, and a Win10 64-bit operating system, and the compilation environment was Dev C++ 5.4.0.

This paper proposes EFO-Miner and OPR-Miner to mine frequent OPPs and strong OPRs, respectively. OPR-Miner adds two branch statements on the basis of EFO-Miner, which hardly takes time. Therefore, we only validate the running performance of EFO-Miner, since the running performance of OPR-Miner is almost the same as EFO-Miner. Moreover, we verify the usefulness of strong OPRs mined by OPR-Miner.


\textbf{For EFO-Miner:}

1) Mat-Based: {To verify the efficiency of EFO-Miner, we developed Mat-Based which employs the pattern fusion strategy to generate candidate patterns and adopts an OPP matching algorithm proposed in \cite{kim2014orde} to calculate the support for each candidate pattern.}

2) OPP-Miner {\cite{wu2021orde}}: To validate the efficiency of EFO-Miner, we selected OPP-Miner as a competitive algorithm. OPP-Miner adopts a pattern matching strategy to calculate the support and needs to scan the sequence numerous times.

3) EFO-enum: To test the performance of the pattern fusion strategy in terms of generating super-patterns, we developed EFO-enum, which employs an enumeration strategy to generate super-patterns and SFP to calculate the support. 

4) EFO-scrn: To verify the effect of the screening strategy on the calculation of supports, we developed EFO-scrn, which does not apply the screening strategy. Since the pruning strategy is based on the screening strategy, EFO-scrn employs neither pruning strategy nor screening strategy, and instead adopts pattern fusion to generate candidate patterns and SFP to calculate the support.

5) EFO-prun: To validate the performance of the pruning strategy, we proposed EFO-prun, which does not apply the pruning strategy, and instead adopts pattern fusion to generate candidate patterns and SFP-Pro to calculate the support.

\textbf{For OPR-Miner:}
6) OPR-Rule: To report the confidences of the strong rules mined by OPR-Miner, we explored OPR-Rule to generate all OPRs based on all frequent OPPs.

\subsection{Performance of EFO-Miner}\label{sub5.2}

{To validate the performance of EFO-Miner, we used five competitive algorithms: Mat-Based, OPP-Miner, EFO-enum, EFO-scrn, and EFO-prun.} We performed experiments on the SDB1–SDB8 datasets, and set the minimum support threshold \textit{minsup} = 12. Since all six algorithms are complete, the mining results are the same, i.e., there are 17, 72, 160, 297, 497, 741, 1162, and 1023 frequent patterns for SDB1–SDB8, respectively. Comparisons of the running time and numbers of candidate patterns are shown in Figs. \ref{figure4} and \ref{figure5}, respectively. We also show a comparison of the numbers of elements in the prefix and suffix arrays in Fig. \ref{figure6} (this figure does not include both Mat-Based and OPP-Miner, since the two algorithms do not use prefix and suffix arrays to calculate the support).

\begin{figure}[!htb]
	\centering
	\includegraphics[width=\linewidth]{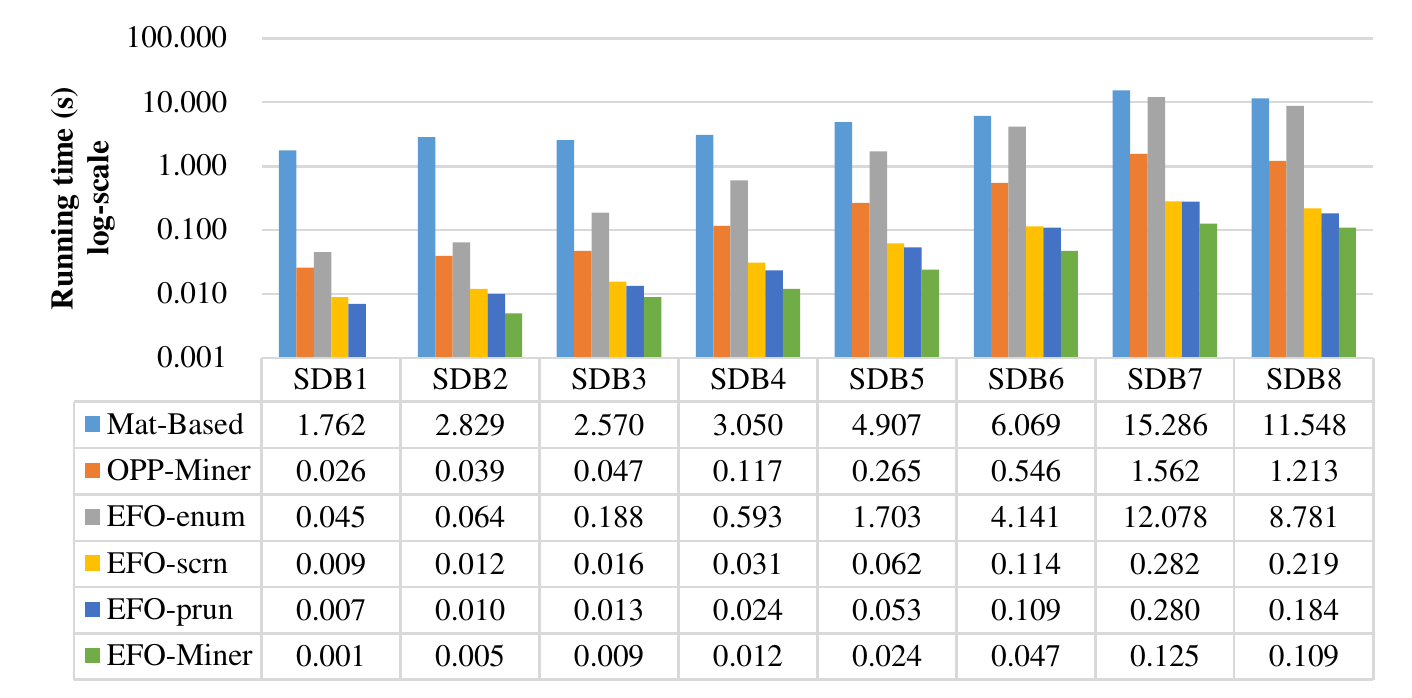}
	\caption{Comparison of running time on SDB1–SDB8}
	\label{figure4}		
\end{figure}

\begin{figure}[!htb]
	\centering
	\includegraphics[width=\linewidth]{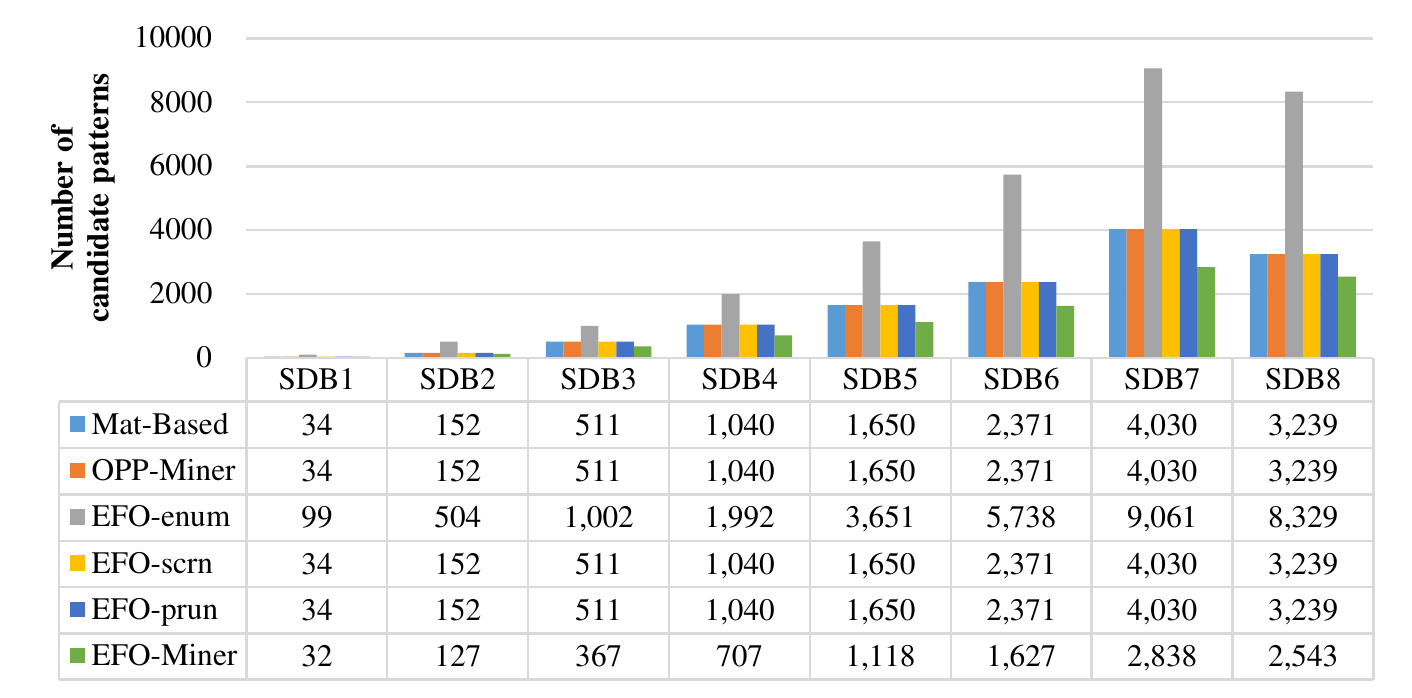}
	\caption{Comparison of numbers of candidate patterns for SDB1–SDB8}
	\label{figure5}		
\end{figure}

\begin{figure}[!htb]
	\centering
	\includegraphics[width=\linewidth]{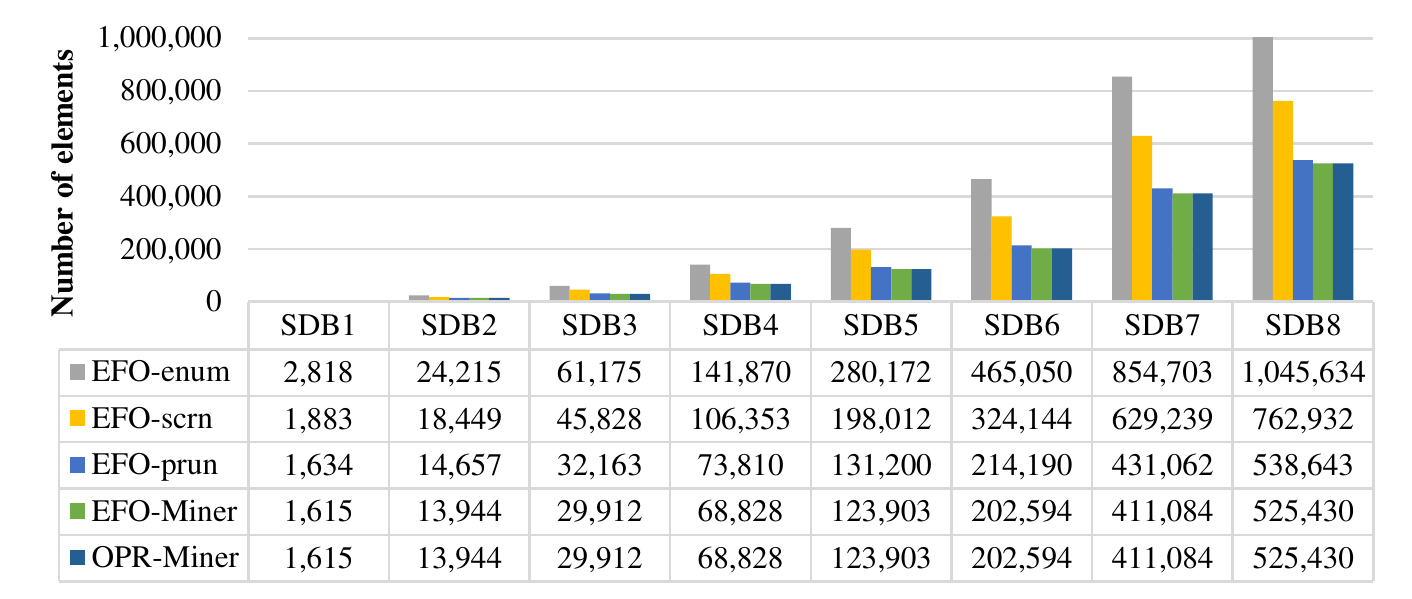}
	\caption{Comparison of numbers of elements in prefix and suffix arrays for SDB1–SDB8}
	\label{figure6}		
\end{figure}

The results give rise to the following observations.

1) EFO-Miner gives better performance than both Mat-Based and OPP-Miner, since EFO-Miner not only runs faster than the two algorithms, but also checks fewer candidate patterns. For example, on SDB7, Fig. \ref{figure4} shows that EFO-Miner takes 625 ms, while Mat-Based and OPP-Miner take 10,421 and 1,266 ms, respectively; Fig. \ref{figure5} shows that EFO-Miner checks 2,838 candidate patterns, while Mat-Based and OPP-Miner check both 4,030. The same effect can be found on all the other datasets. The reasons for this are as follows.  Mat-Based and OPP-Miner employ different pattern matching strategies that cannot use the results for the sub-patterns and has to scan the database repeatedly, which is inefficient. In contrast, EFO-Miner uses the results for the sub-patterns to calculate the occurrences of super-patterns, which can avoid redundant calculations and improve the efficiency. Moreover, although EFO-Miner, Mat-Based, and  OPP-Miner adopt a pattern fusion strategy to generate candidate patterns, EFO-Miner employs a pruning strategy that can further reduce the number of candidate patterns. Hence, EFO-Miner checks fewer candidate patterns than both Mat-Based and OPP-Miner, and therefore outperforms them.

2) EFO-Miner outperforms EFO-enum, thus demonstrating that the pattern fusion strategy can efficiently prune candidate patterns. Fig. \ref{figure4} shows that EFO-Miner runs faster than EFO-enum. For example, on SDB4, EFO-Miner takes 65.6 ms, while EFO-enum takes 2,594 ms. The same effect can be found on all the other datasets. The reason for this is that the pattern fusion strategy can effectively reduce the number of candidate patterns. For example, Fig. \ref{figure5} shows that on SDB4, EFO-Miner generates 707 candidate patterns, while EFO-enum generates 1,992. From Fig. \ref{figure6}, we can see that on SDB4, EFO-Miner carries out 68,828 comparisons between elements, while for EFO-enum it is 141,870. The experimental results are therefore consistent with those in Example \ref{example4}. We know that the lower the number of candidate patterns, the faster the algorithm runs. Hence, EFO-Miner runs faster than EFO-enum.

3) EFO-Miner outperforms EFO-prun, which indicates that the pruning strategy can efficiently reduce the number of candidate patterns. Fig. \ref{figure4} shows that EFO-Miner runs faster than EFO-prun. For example, on SDB5, EFO-Miner takes 125 ms, while EFO-prun takes 234 ms, and the same effect can be found on the other datasets. The reason for this is that the pruning strategy can effectively reduce the number of candidate patterns. For example, from Fig. \ref{figure5}, we can see that EFO-Miner checks 1,118 candidate patterns for SDB5, while EFO-prun checks 1,650. From Figure \ref{figure6}, we see that on SDB5, EFO-Miner carries out 123,903 comparisons between elements, while for EFO-scrn it is 131,200. With a reduction in the number of candidate patterns, the number of comparisons is also reduced. The experimental results are therefore consistent with those in Example \ref{example7}. We know that the lower the number of candidate patterns, the faster the algorithm runs. EFO-Miner therefore runs faster than EFO-prun. 

4) EFO-Miner outperforms EFO-scrn. More importantly, EFO-prun outperforms EFO-scrn, which indicates that the screening strategy can efficiently improve the mining performance. Fig. \ref{figure4} shows that EFO-prun runs faster than EFO-scrn. For example, on SDB3, EFO-prun takes 37.3 ms, while EFO-scrn takes 44.8 ms, and the same effect can be found on all the other datasets. The reason for this is that the screening strategy can dynamically reduce the size of the prefix and suffix arrays. For example, Fig. \ref{figure6} shows that on SDB3, EFO-prun carries out 32,163 comparisons between elements, while EFO-scrn carries out 45,828. The experimental results are therefore consistent with those in Example \ref{example6}. The lower the sizes of the prefix and suffix arrays, the faster the algorithm runs, meaning that EFO-prun runs faster than EFO-scrn. We know that EFO-Miner runs faster than EFO-prun. Hence, EFO-Miner runs faster than EFO-scrn.


\subsection{{Scalability}}\label{scalabilty}

{In this section, to evaluate the scalability of EFO-Miner, we employed Mat-Based, EFO-enum, EFO-scrn, and EFO-prun as competitive algorithms. Moreover, we selected SDB8 as the experimental dataset, and created SDB8\_1, SDB8\_2, SDB8\_3, SDB8\_4, SDB8\_5, and SDB8\_6, which are one, two, three, four, five, and six times the size of SDB8, respectively. Obviously, if minsup is a constant, the longer the sequence, the more frequent patterns. The running time is positive related with number of frequent patterns and the sequence length according to Theorem 7. To avoid the impact of the different number of patterns on the running time, we set $minsup$=10, 20, 30, 40, 50, and 60 on SDB8\_1-SDB8\_6. All these algorithms mine 1243 patterns, and the comparison of running time is shown in Fig. \ref{Scal}.}

\begin{figure}[!htb]
	\centering
	\includegraphics[width=\linewidth]{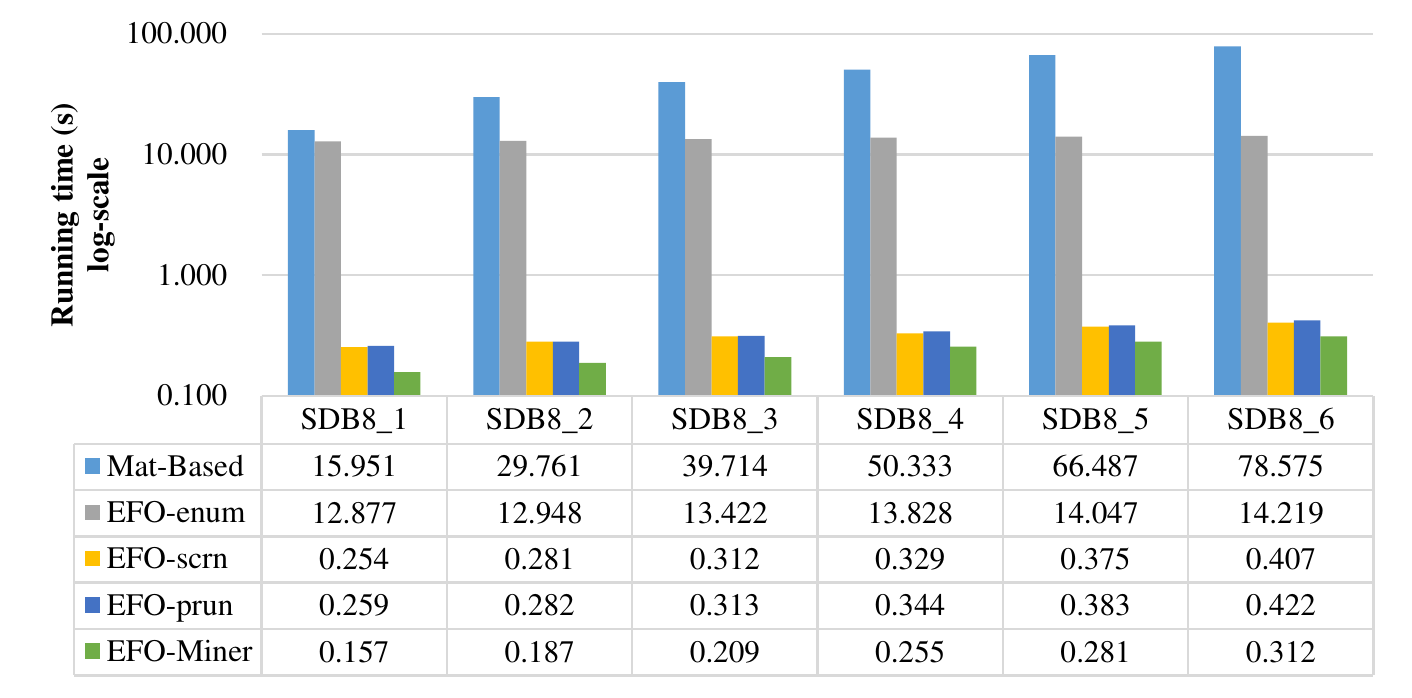}
	\caption{{Comparison of running time with different dataset sizes }}
	\label{Scal}		
\end{figure}

The results give rise to the following observations. From Fig. \ref{Scal}, we know that the running time of EFO-Miner grows slower than the dataset size. For example, the size of SDB8\_6 is six times of SDB8\_1, while EFO-Miner takes 0.331s on SDB8\_6, which is 0.331/0.147=2.252 times of SDB8\_1. This phenomenon can be found in all other datasets. The results indicate that the running time is positively correlated with the dataset size, which is consistent with the analysis of time complexity of EFO-Miner. More importantly, EFO-Miner runs significantly faster than other competitive algorithms, such as Mat-Based, EFO-enum, EFO-scrn, and EFO-prun. The reason is the same as the analysis in Section \ref{sub5.2}. Hence, EFO-Miner has strong scalability, since the mining performance does not degrade as the dataset size increases.

\subsection {{Influence of different \textit{minsup}}}\label {minsup} 

{In this section, to report the influence of different \textit{minsup} on number of patterns and running time of EFO-Miner, we selected Mat-Based, EFO-enum, EFO-scrn, and EFO-prun as competitive algorithms, and selected dataset SDB16 and expanded it by 10 times to obtain a larger dataset as the experimental dataset. We set \textit{minsup}=1600, 1700, 1800, 1900, 2000, and 2100, respectively. The comparison of number of patterns and running time on SDB16 are shown in Figs. \ref {sup_pattern} and \ref{sup_time}, respectively.}

\begin{figure}[!htb]
	\centering
	\includegraphics[width=\linewidth]{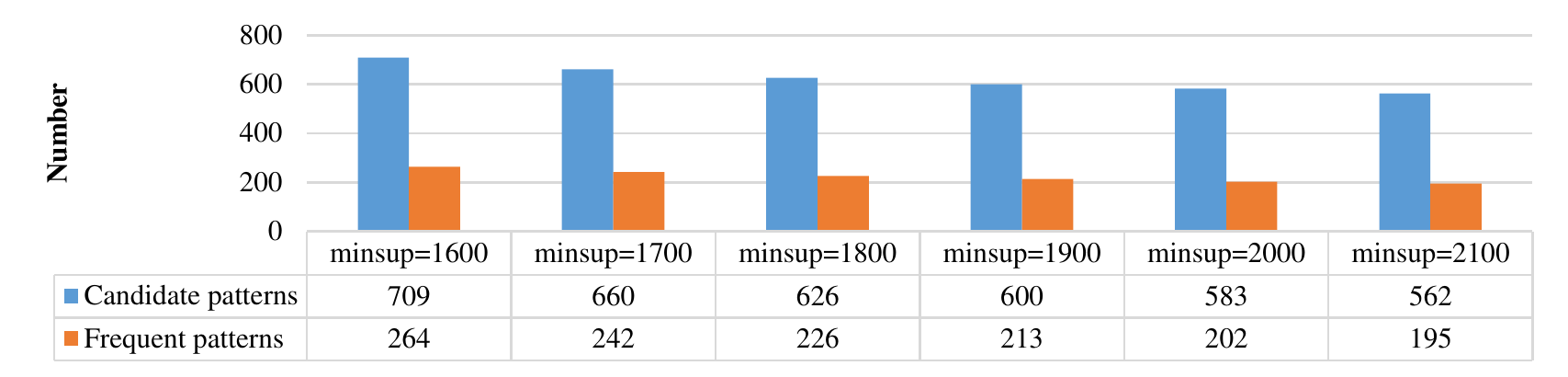}
	\caption{{Comparison of number of patterns with different \textit{minsup} on SDB16 }}
	\label{sup_pattern}		
\end{figure}

\begin{figure}[!htb]
	\centering
	\includegraphics[width=\linewidth]{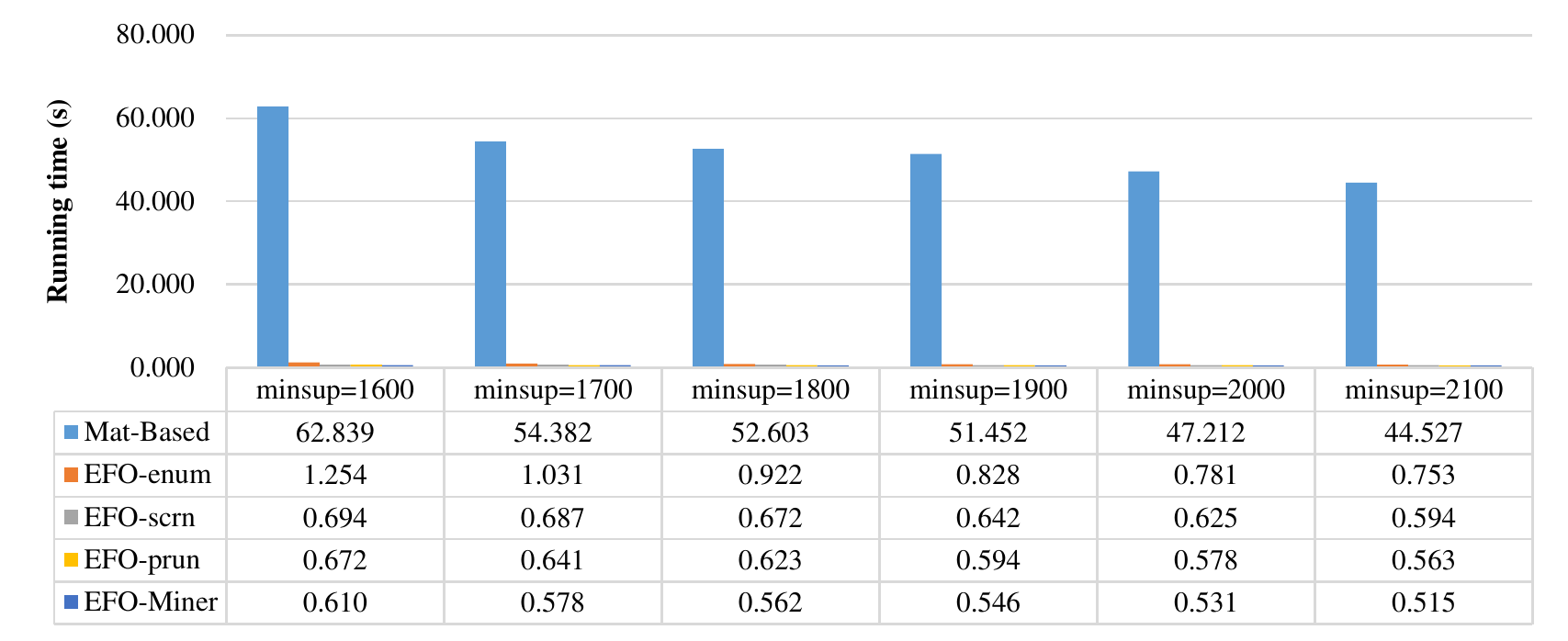}
	\caption{{Comparison of running time with different \textit{minsup} on SDB16 }}
	\label{sup_time}		
\end{figure}

{The results give rise to the following observations. With the increase of \textit{minsup}, the number of patterns and running time decreases. For example, from Figs. \ref {sup_pattern} and  \ref{sup_time}, when \textit{minsup}=1600, EFO-Miner discovers 264 OPPs and takes 0.610s, whereas when \textit{minsup}=2100, EFO-Miner discovers 195 OPPs and takes 0.515s. This phenomenon can also be found in other competitive algorithms. The reason for this is as follows. With the increase of \textit{minsup} value, the number of frequent patterns decreases. As a result, the running time also decreases.   Moreover, EFO-Miner outperforms other competitive algorithms, which is consistent with the results of Section \ref{sub5.2}.}

\subsection{Performance of OPR-Miner}\label{sub5.3}

In this case, OPR-Rule was selected as a comparison algorithm to generate all the OPRs, and experiments were carried out on SDB1–SDB8. We set the minimum support threshold \textit{minsup}=12 and the minimum confidence threshold \textit{minconf}=0.45. The number of generated rules is shown in Fig. \ref{figure8}. Moreover, Fig. \ref{figure9} shows the comparison of the confidences of OPRs and strong OPRs for SDB3.

\begin{figure}[!htb]
	\centering
	\includegraphics[width=\linewidth]{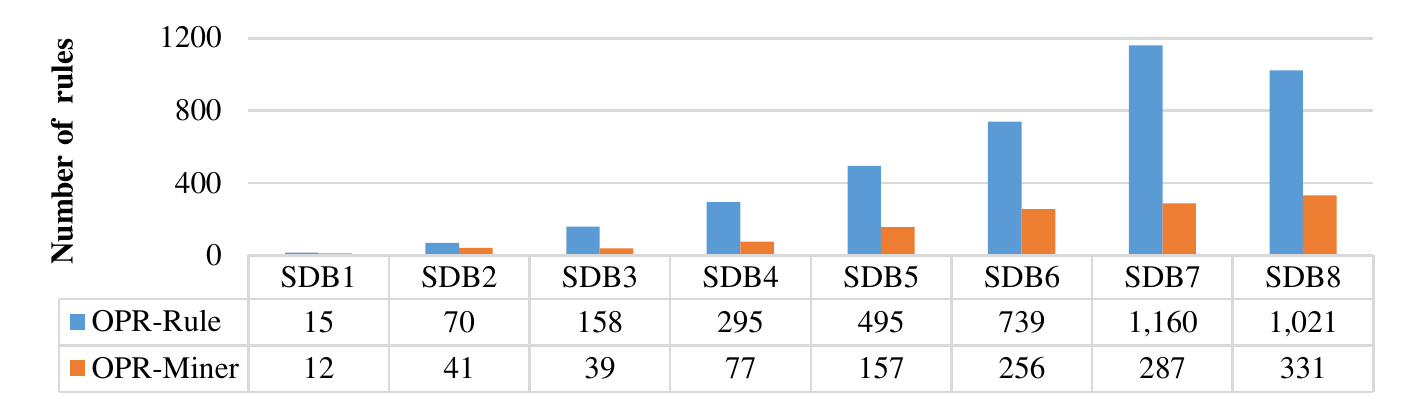}
	\caption{Comparison of number of rules on SDB1–SDB8}
	\label{figure8}		
\end{figure}

\begin{figure}[!htb]
	\centering
	\includegraphics[width=\linewidth]{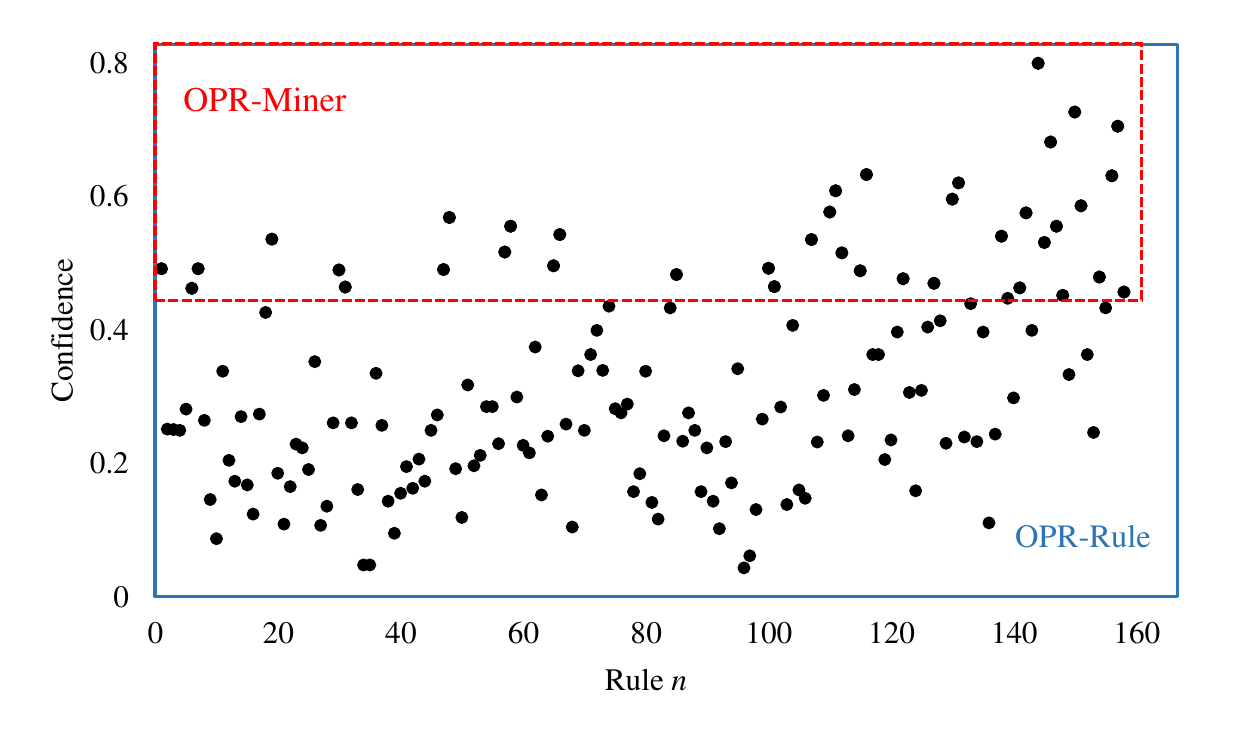}
	\caption{Comparison of confidences of OPRs and strong OPRs for SDB3. A point represents a rule, the X-axis represents the $n$-th rule, and the Y-axis represents the confidence level of the rule. }
	\label{figure9}		
\end{figure}

The results indicate that OPR-Miner outperforms OPR-Rule, thus validating that OPR-Miner can efficiently prune rules.  For example, from Fig. \ref{figure8}, we know that on SDB3, OPR-Miner generates 39 candidate patterns, while OPR-Rule generates 158. Our experimental results are therefore consistent with those in Example \ref{example9}. Moreover, Fig. \ref{figure9} shows that the mined rules of OPR-Miner are a part of OPR-Rule. More importantly, OPR-Miner can mine rules with high confidences. Since we set the minimum confidence threshold \textit{minconf} = 0.45, the confidences of the OPRs mined by OPR-Miner are no less than 0.45, while some of the confidences of OPRs mined by OPR-Rule are less than 0.45. Hence, OPR-Miner can find more useful rules than OPR-Rule.

\subsection{{Influence of different \textit{minconf}}}\label{subconf} 

{To report the influence of different \textit {minconf} on the number of patterns and running time of OPR-Miner, we also selected dataset SDB16 and expanded it by 10 times to obtain a larger dataset as the experimental dataset. We selected OPR-Rule as the competitive algorithm. We set \textit {minsup}=1800 and \textit {minconf}=0.40, 0.45, 0.50, 0.55, 0.60, and 0.65, respectively. The running time of OPR-Miner and OPR-Rule on all \textit {minconf} is all about 0.563s, and the comparison of number of strong OPRs with different \textit {minconf} is shown in Fig. \ref{differentconf}.}

\begin{figure}[!htb]
	\centering
	\includegraphics[width=\linewidth]{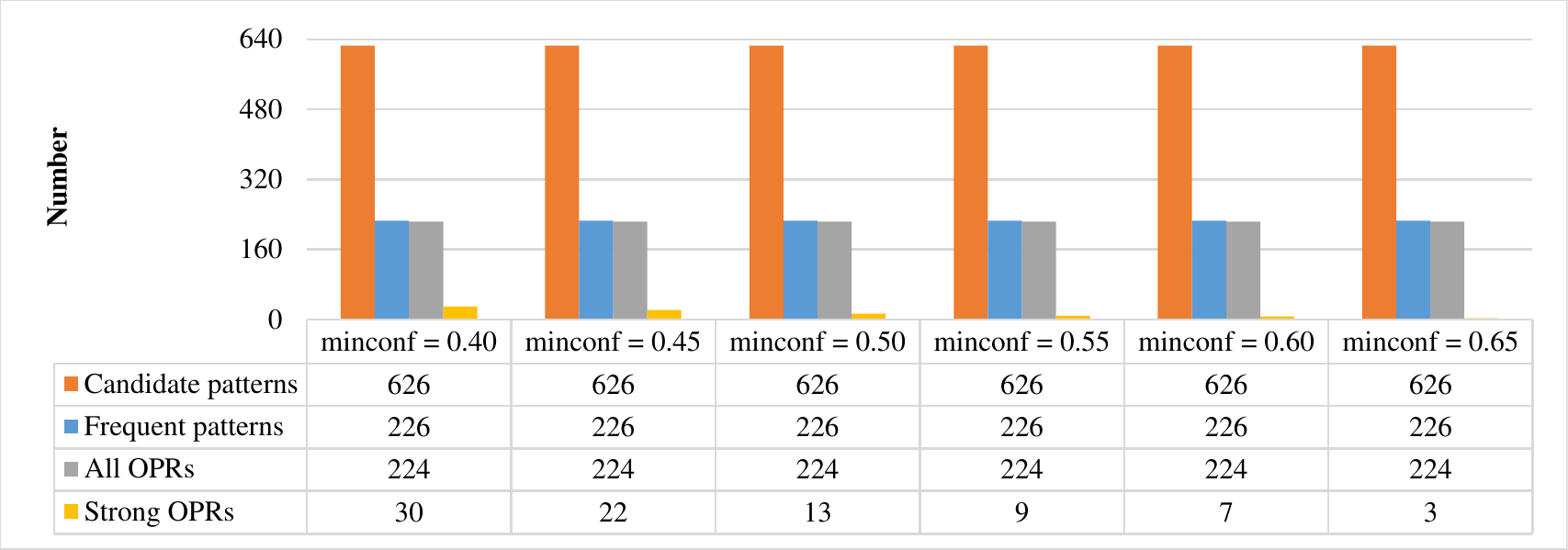}
	\caption{Comparison of number of strong OPRs with different \textit {minconf}}
	\label{differentconf}		
\end{figure}

The results give rise to the following observations. The running time of OPR-Miner and OPR-Rule are almost the same, since OPR-Miner discovers a subset of OPR-Rule, and the process requires almost no time. Moreover, with the increase of \textit {minconf}, the number of candidate patterns, frequent patterns, and all OPRs are constant, while the number of strong OPRs decreases. For example, no matter what \textit {minconf} is, the number of candidate patterns, frequent patterns, and all OPRs are 626, 226, and 224, respectively. However, when \textit {minconf}=0.40, the number of strong OPRs is 30, while when \textit {minconf}=0.65, the number of strong OPRs is 3. The reason is as follows. We know that the number of candidate patterns, frequent patterns, and all OPRs are not related with the value of \textit {minconf}. Thus, with the increase of \textit {minconf}, the number of candidate patterns, frequent patterns, and all OPRs are constant. However, the number of strong OPRs is related with the value of \textit {minconf}. With the increase of \textit {minconf}, the number of strong OPRs decreases.

\subsection{Case studies}\label{sub5.4}

In this section, we report the applications of rule mining from two aspects: clustering and classification. To evaluate the performance of OPR-Miner, we selected OPP-Miner as the competitive model. We know that each OPR can be represented by \textbf{x}$\to$\textbf{y}, where \textbf{x} and \textbf{y} are two frequent OPPs. Thus, each OPR is composed by two OPPs. If we mine \textit{t} OPRs, then there will be \textit{k} different OPPs, where \textit{k}$\leq2\times$ \textit{t}, since some OPPs may be the same. For example, (1,2,3)$\to$(1,2,3,4) and (1,2,3,4)$\to$(1,2,3,4,5) are two different OPRs. However, there are only three different OPPs: (1,2,3), (1,2,3,4), and (1,2,3,4,5). We use \textit{k} supports of corresponding OPPs to form a new dataset as OPRs. For fairness, we also selected top-\textit{k} supports of OPPs to form a new dataset as OPPs. The original dataset is called Raw.

\subsubsection{Clustering performance} 

To validate the clustering performance, a clustering experiment is conducted in this section. Since SDB1-SDB8 are single sequence datasets, clustering experiment cannot be conducted. Thus, we selected SDB9-SDB12 to conduct the experiment according to the following steps.

1. We employ OPR-Miner to mine OPRs and the parameters are \textit{minsup} = 25 and \textit{minconf} = 0.65. We discover 6 OPRs corresponding to 8 OPPs on SDB9, 8 OPRs corresponding to 12 OPPs on SDB10, 8 OPRs corresponding to 12 OPPs on SDB11, and 8 OPRs corresponding to 11 OPPs on SDB12. Therefore, we discover top-8, top-12, top-12, and top-11 OPPs on SDB9, SDB10, SDB11, and SDB12, respectively. {We show the comparison of the mined OPRs, their corresponding OPPs, top-$k$ OPPs, and shared OPPs in Table \ref {tableopp}.}

\begin{table}
	\renewcommand{\arraystretch}{1.15}
	\tiny
	\caption{{Comparison of mined patterns}}
	\label{tableopp}
	\centering 
	\begin{threeparttable}
		\resizebox{\linewidth}{!}{
			\setlength{\tabcolsep}{0.65mm}{
				\begin{tabular}{|c|c|c|c|}
					\hline
					Dataset & Type 	&  Number & 	Mined OPRs or OPPs \\ \hline
					\multirow{5}*{SDB9} & \multirow{2}*{Strong OPRs} & \multirow{2}*6 & (1,2)$\to$(1,2,3), (2,1)$\to$(3,2,1), (1,2,3)$\to$(1,2,3,4),\\&&& (3,2,1)$\to$(4,3,2,1), (1,2,3,4)$\to$(1,2,3,4,5), (1,2,3,4,5)$\to$(1,2,3,4,5,6)\\ \cline{2-4}
					~ & Corresponding OPPs &8 & (1,2), (2,1), (1,2,3), (3,2,1), (1,2,3,4), (4,3,2,1), (1,2,3,4,5), (1,2,3,4,5,6) \\ \cline{2-4}
					~ & Top-$k$ OPPs & 8 & (1,2), (1,2,3), (2,1), (1,2,3,4), (1,2,3,4,5), (3,2,1), (1,2,3,4,5,6), (1,2,3,4,5,6,7) \\ \cline{2-4}
					~ & Shared OPPs    & 7 &(1,2), (2,1), (1,2,3), (3,2,1), (1,2,3,4), (1,2,3,4,5), (1,2,3,4,5,6)\\ 
					\hline
					
					\multirow{9}*{SDB10} &&& (1,2,3)$\to$(1,2,3,4), (3,2,1)$\to$(4,3,2,1), (1,2,3,4)$\to$(1,2,3,4,5), 
					\\& Strong OPRs &8 &(4,3,2,1)$\to$(5,4,3,2,1), (1,2,3,4,5)$\to$(1,2,3,4,5,6), (5,4,3,2,1)$\to$(6,5,4,3,2,1), 
					\\&&&(1,2,3,4,5,6)$\to$(1,2,3,4,5,6,7), (6,5,4,3,2,1)$\to$(7,6,5,4,3,2,1) \\ \cline{2-4} 
					& \multirow{2}*{Corresponding OPPs} & \multirow{2}*{12} & (1,2), (2,1), (1,2,3), (3,2,1), (1,2,3,4), (4,3,2,1), (1,2,3,4,5), 
					\\&&&(5,4,3,2,1), (1,2,3,4,5,6), (6,5,4,3,2,1), (1,2,3,4,5,6,7), (7,6,5,4,3,2,1)\\ \cline{2-4} 
					& \multirow{2}*{Top-$k$ OPPs} & \multirow{2}*{12} & (2,1), (3,2,1), (1,2), (4,3,2,1), (1,2,3), (5,4,3,2,1), (6,5,4,3,2,1), (1,2,3,4),
					\\&&&(7,6,5,4,3,2,1), (1,2,3,4,5), (8,7,6,5,4,3,2,1), (9,8,7,6,5,4,3,2,1) \\ \cline{2-4} 
					& \multirow{2}*{Shared OPPs}   &\multirow{2}*{10}& (2,1), (3,2,1), (1,2), (4,3,2,1), (1,2,3), (5,4,3,2,1), \\&&&(6,5,4,3,2,1), (1,2,3,4), (7,6,5,4,3,2,1), (1,2,3,4,5)  \\ \hline			
					
					\multirow{9}*{SDB11} & & & (1,2,3)$\to$(1,2,3,4), (3,2,1)$\to$(4,3,2,1), (1,2,3,4)$\to$(1,2,3,4,5), 
					\\&Strong OPRs & 8&(4,3,2,1)$\to$(5,4,3,2,1), (1,2,3,4,5)$\to$(1,2,3,4,5,6), (5,4,3,2,1)$\to$(6,5,4,3,2,1), 
					\\&&&(1,2,3,4,5,6)$\to$(1,2,3,4,5,6,7), (6,5,4,3,2,1)$\to$(7,6,5,4,3,2,1)\\ \cline{2-4}
					~ & \multirow{2}*{Corresponding OPPs} &\multirow{2}*{12} & (1,2,3), (1,3,2), (2,1,3), (3,2,1), (1,2,3,4), (4,3,2,1), (1,2,3,4,5),\\&&& (5,4,3,2,1), (1,2,3,4,5,6), (6,5,4,3,2,1), (1,2,3,4,5,6,7), (7,6,5,4,3,2,1)\\ \cline{2-4}
					~ & & & (2,1), (3,2,1), (4,3,2,1), (5,4,3,2,1), (6,5,4,3,2,1), (7,6,5,4,3,2,1),
					\\&Top-$k$ OPPs & 12& (8,7,6,5,4,3,2,1), (9,8,7,6,5,4,3,2,1), (10,9,8,7,6,5,4,3,2,1), 
					\\&&&(11,10,9,8,7,6,5,4,3,2,1), (12,11,10,9,8,7,6,5,4,3,2,1), (1,2) \\ \cline{2-4}
					~ & Shared OPPs    & 5 & (3,2,1), (4,3,2,1), (5,4,3,2,1), (6,5,4,3,2,1), (7,6,5,4,3,2,1)\\ 
					\hline
					
					\multirow{9}*{SDB12} & & &  (2,1)$\to$(3,2,1), (1,2,3)$\to$(1,2,3,4), (3,2,1)$\to$(4,3,2,1), 
					\\&Strong OPRs &8&(1,2,3,4)$\to$(1,2,3,4,5), (4,3,2,1)$\to$(5,4,3,2,1), (1,2,3,4,5)$\to$(1,2,3,4,5,6),
					\\&&& (5,4,3,2,1)$\to$(6,5,4,3,2,1), (1,2,3,4,5,6)$\to$(1,2,3,4,5,6,7)\\ \cline{2-4} 
					& \multirow{2}*{Corresponding OPPs} & \multirow{2}*{11} & (1,2), (2,1), (1,2,3), (3,2,1), (1,2,3,4), (4,3,2,1), (1,2,3,4,5), \\&&&(5,4,3,2,1), (1,2,3,4,5,6), (6,5,4,3,2,1), (1,2,3,4,5,6,7)\\ \cline{2-4} 
					& \multirow{2}*{Top-$k$ OPPs} & \multirow{2}*{11} & (2,1), (1,2), (3,2,1), (1,2,3), (4,3,2,1), (5,4,3,2,1), (1,2,3,4), \\&&& (6,5,4,3,2,1), (1,2,3,4,5), (7,6,5,4,3,2,1), (1,2,3,4,5,6) \\ \cline{2-4} 
					& \multirow{2}*{Shared OPPs}   &\multirow{2}*{10}& (1,2), (2,1), (1,2,3), (3,2,1), (1,2,3,4), (4,3,2,1), \\&&&(1,2,3,4,5), (5,4,3,2,1), (1,2,3,4,5,6), (6,5,4,3,2,1) \\ \hline

				\end{tabular}
				
			}
		}
	\end{threeparttable}
\end{table}

2. We adopt K-Means to cluster the Raw, OPPs, and OPRs data with parameter \textit{K} = 7. 

3. To evaluate the clustering performance, we select two criteria: Normalized Mutual Information (\textit{NMI}) {\cite{danon2005comp}} and Homogeneity (\textit{h}) {\cite{rosenberg2007vmea}}, which can be calculated according to Equations \ref{eq1} and \ref{eq2}, respectively.


\begin{equation}\label{eq1}
NMI(X,Y)=\frac{\sum_{i=1}^{|X|} \sum_{j=1}^{|Y|} P(i, j) \log \left(\frac{P(i, j)}{P(i) P(j)}\right)}{\sqrt{\sum_{i=1}^{|X|} P(i) \log P(i) \times \sum_{j=1}^{|Y|} P(j) \log P(j)}}	
\end{equation}

\begin{equation}\label{eq2}
h(X, Y)=1-\frac{-\sum_{i=1}^{|X|} \sum_{j=1}^{|Y|} P(i, j) \log P(i \mid j)}{-\sum_{i=1}^{|X|} P(i) \log P(i)} 
\end{equation}

Both \textit{NMI} and \textit{h} reflect the similarity between the clustering results and the actual values. {The greater the \textit{NMI} and \textit{h}, the greater the similarity, i.e., the better the clustering performance.} The comparison of clustering performances is shown in Fig. \ref{figure10}.

\begin{figure}[!htb]
	\centering
	\includegraphics[width=\linewidth]{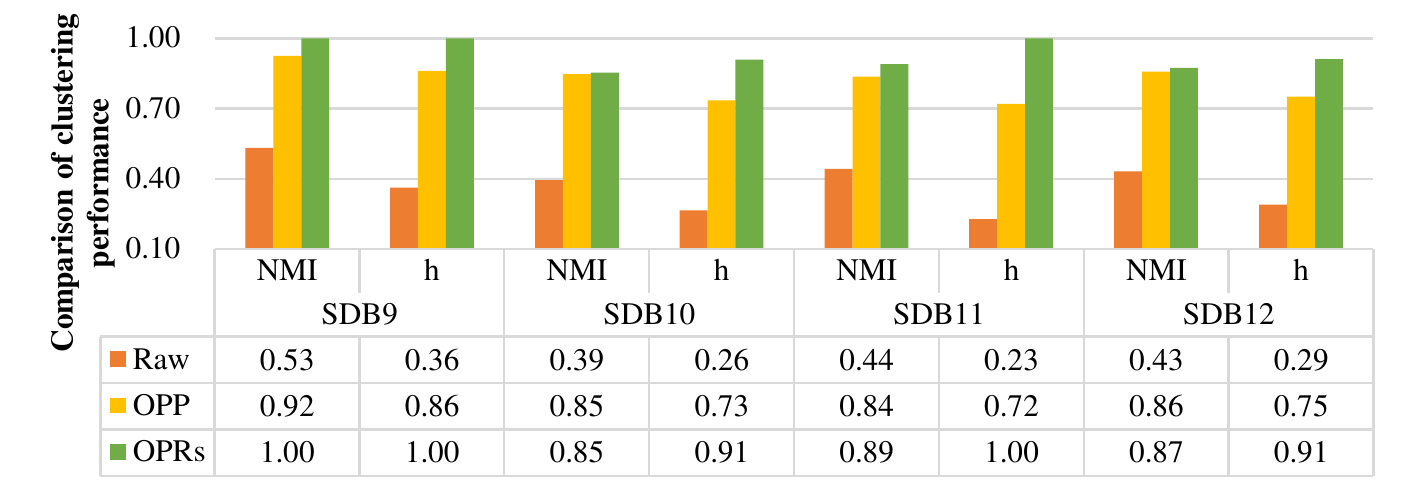}
	\caption{Comparison of clustering performances on SDB9–SDB12}
	\label{figure10}		
\end{figure}

The results give rise to the following observations.

1.  Both OPP-Miner and OPR-Miner can effectively extract the critical information from the original time series. For example, the values of \textit{NMI} of Raw, OPPs, and OPRs on SDB9 are 0.57, 0.85, and 0.89, respectively. The performances of OPP-Miner and OPR-Miner are better than Raw. The same effect can be found on all other datasets. The reason is that the original data may contain much redundant information, which can affect the clustering performance, while both OPP-Miner and OPR-Mine use the frequent trends to represent the original time series, which are more critical information with high support and high confidence. The results indicate that OPP-Miner and OPR-Miner can be used for feature selection for clustering task.

2. OPR-Miner has better performance than OPP-Miner. For example, the values of \textit{NMI} of OPPs and OPRs on SDB10 are 0.86 and 0.88, respectively. The same effect can be found on all other datasets except SDB9 for \textit{h}. The reason is that although top-\textit{k} OPPs are very critical information with high supports, some OPPs have lower confidence. However, OPR-Miner can extract the critical information with high support and high confidence, which can improve the clustering performance.

3.  { It is a very interesting phenomenon that some datasets share many common OPPs, while others share few. For example, on SDB9, OPR-Miner discovers six strong OPRs which are composed of eight patterns, and among them, seven patterns are Top-8 OPPs. But on SDB11, OPR-Miner discovers eight strong OPRs which are composed of 12 patterns, and among them, only five patterns are Top-12 OPPs. This result indicates that there is no clear relationship between top-$k$ OPPs and strong OPRs. For a specific time series clustering problem, how to extract effective features to achieve high-quality clustering performance is worthy of further study.}

\subsubsection{Classification performance}

To validate the classification performance, a classification experiment is conducted in this section. We conducted the experiment on SDB13-SDB15. We chose five classical classification algorithms: SVM with Polynomial kernel function, C4.5, CART, AdaBoost, and KNN, which are Top 10 algorithms in data mining {\cite{wu2008top1}}.

To mine OPRs, the parameters are \textit{minsup} = 15 and \textit{minconf} = 0.25. We discover 7 OPRs corresponding to 12 OPPs on SDB13, 5 OPRs corresponding to 10 OPPs on SDB14, and 7 OPRs corresponding to 10 OPPs on SDB15. Since the three datasets are binary classification datasets, we adopt the prediction accuracy as the criterion. Moreover, we employ three-fold cross-validation to verify the classification performance. The comparisons of accuracy on SDB13, SDB14, and SDB15 are shown in Figs. \ref{figure11}, \ref{figure12}, and \ref{figure13}, respectively.

\begin{figure}[!htb]
	\centering
	\includegraphics[width=\linewidth]{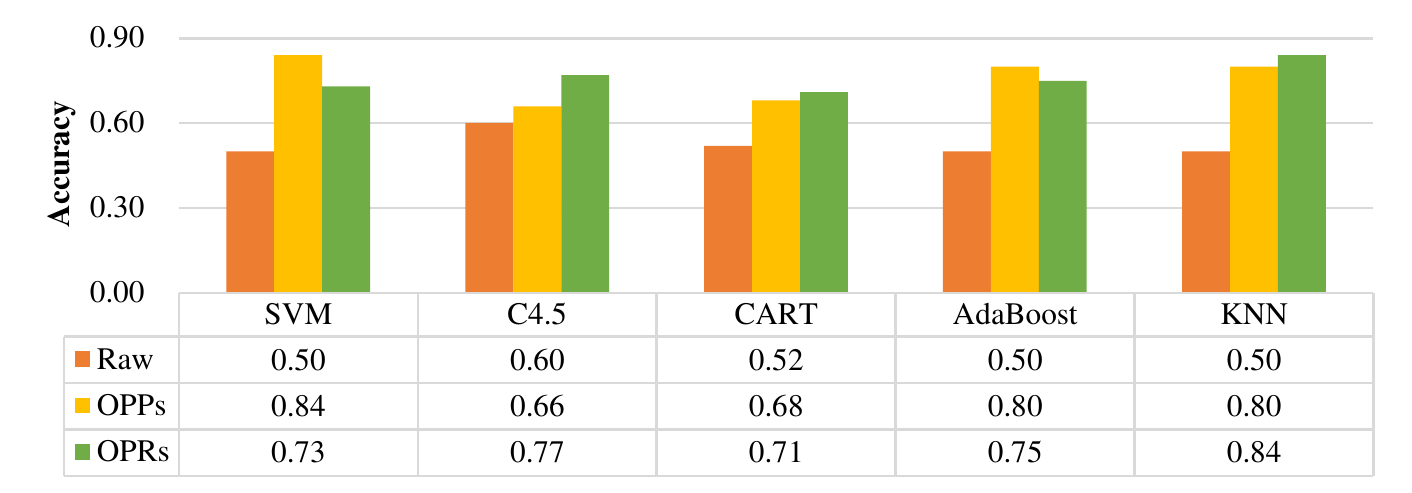}
	\caption{Comparison of accuracy on SDB13}
	\label{figure11}		
\end{figure}
\begin{figure}[!htb]
	\centering
	\includegraphics[width=\linewidth]{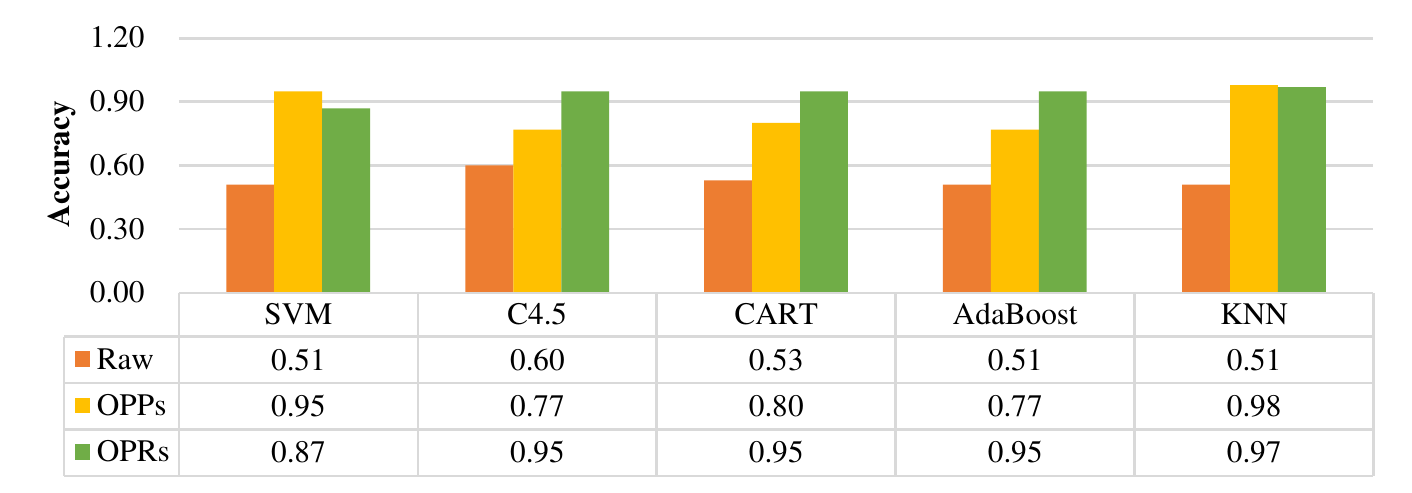}
	\caption{Comparison of accuracy on SDB14}
	\label{figure12}		
\end{figure}
\begin{figure}[!htb]
	\centering
	\includegraphics[width=\linewidth]{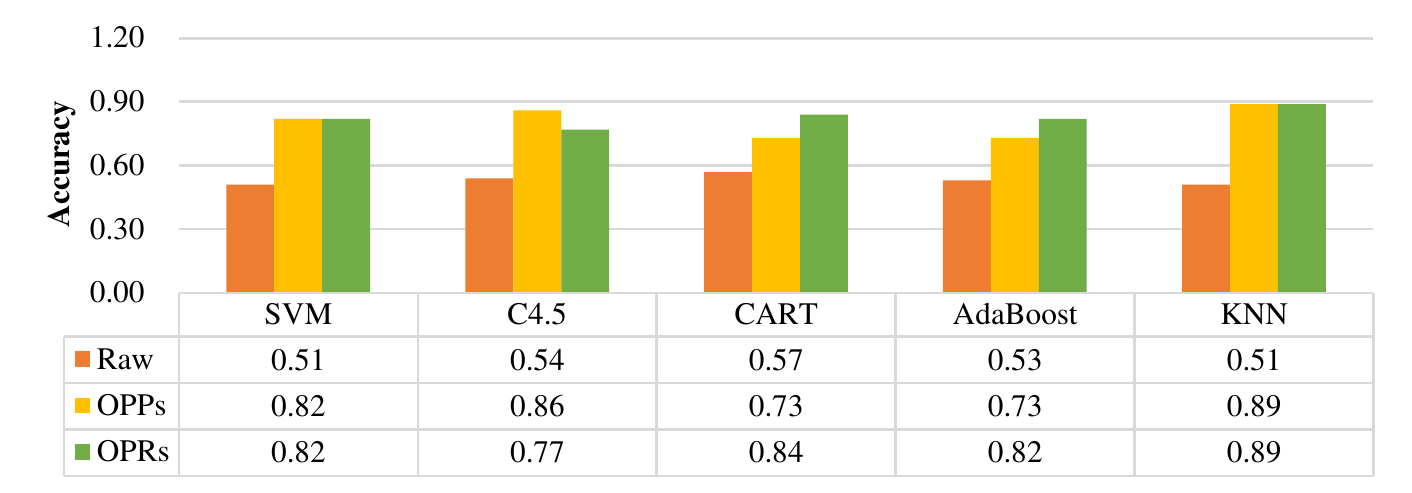}
	\caption{Comparison of accuracy on SDB15}
	\label{figure13}		
\end{figure}

From Figs. \ref{figure11}-\ref{figure13}, we observe that both OPP-Miner and OPR-Miner can effectively improve the classification performance. For example, in Fig. \ref{figure11}, if we use C4.5 as the classifier, the accuracy of the original data on SDB13 is 0.60, while those of OPPs and OPRs are 0.66 and 0.77, respectively. The classification performance is significantly improved. This effect can be found on all the other datasets. Moreover, OPR-Miner has better classification performance than OPP-Miner. The reason is the same as that in clustering experiments.

\section{Conclusion}\label{section6}

To improve the efficiency of OPP mining and mine the implicit relations between OPPs, we have addressed the issue of OPR mining and proposed an effective mining algorithm called OPR-Miner. In this approach, the key step is finding frequent OPPs. To mine these frequent OPPs, we proposed an algorithm called EFO-Miner consisting of four parts. To reduce the number of candidate patterns, EFO-Miner adopts a pattern fusion strategy to generate candidate patterns. Moreover, to calculate the supports of super-patterns, EFO-Miner uses the matching results of sub-patterns based on the pattern fusion strategy. To improve the efficiency of support calculations, EFO-Miner employs a screening strategy to dynamically reduce the size of the matching results for sub-patterns. To avoid useless support calculations, EFO-Miner applies a pruning strategy to dynamically prune the sub-patterns for which the size of the matching results is less than the minimum support threshold. Experimental results from weather, oil, and stock datasets verify that OPR-Miner gives better performance than other competitive algorithms. More importantly, clustering and classification experiments validate that OPR-Miner can be used to realize feature extraction and achieve good performance.

\section*{Acknowledgement}
This work was partly supported by National Natural Science Foundation of China (61976240, 52077056, 62120106008),  National Key Research and Development Program of China (2016YFB1000901), and Natural Science Foundation of Hebei Province, China (Nos. F2020202013, E2020202033).

{}
\end{document}